\documentclass[aps,pra,amsmath,amssymb,showpacs,twocolumn,nofootinbib]{revtex4-1}
\usepackage{graphicx,epsfig,epsf}% Include  files
\usepackage{dcolumn}% Align table columns on decimal point
\usepackage{bm}% bold math
\usepackage{mathbbol}
\usepackage{epstopdf}
\usepackage{mathtools}
\usepackage{comment}
\usepackage{graphicx}
\usepackage{dsfont}
\usepackage{amsfonts}
\usepackage{bbm}
\usepackage{xcolor}
\usepackage{stackrel}
\usepackage{multirow}
\usepackage{amsthm,bbm}
\makeatletter

\begin{document}

% other
\newcommand{\bk}{\mathbf{k}}
\newcommand{\br}{\mathbf{r}}
\newcommand{\eps}{\boldsymbol{\varepsilon}}
\newcommand{\epsk}{\boldsymbol{\varepsilon}_\mathbf{k}}
\def\keps{\mathbf{k}\boldsymbol{\varepsilon}}

\newcommand{\matgamma}{\boldsymbol{\gamma}}
\newcommand{\matDelta}{\boldsymbol{\Delta_{\mathrm{dd}}}}

\newcommand{\ga}{\Gamma_1}
\newcommand{\gb}{\Delta \gamma}

\newcommand{\kJ}{k_J}
\newcommand{\PI}{{_\mathrm{PI}}}

\newcommand{\tr}{{\rm Tr}}
\newcommand*{\etal}{\textit{et al.}}
\def\vec#1{\mathbf{#1}}
\def\ket#1{\left|#1\right\rangle}
\def\bra#1{\langle#1|}
\def\dm{\boldsymbol{\wp}}
\def\ldx{\eta_x}

% Correlations functions
\newcommand{\cor}{\mathcal{C}_{ij}(t)}
\newcommand{\corcl}{\mathcal{C}_{ij}^{\mathrm{\,cl}}(t)}
\newcommand{\corem}{\mathcal{C}_{ij}^{\mathrm{em}}(t)}
\newcommand{\cork}{\mathcal{C}_{ij}(\mathbf{k},t)}
\newcommand{\corclk}{\mathcal{C}^\mathrm{\,cl}_{\mathbf{k}\eps} (t)}
\newcommand{\coremk}{\mathcal{C}^\mathrm{em}_{\mathbf{k}\eps} (t)}
\newcommand{\corkm}{\mathcal{C}_{ij}^\mathrm{ex}(\mathbf{k},t)}
\newcommand{\corkmz}{\mathcal{C}_{ij}^\mathrm{ex}(\mathbf{k})}
\newcommand{\cormz}{\mathcal{C}_{ij}^\mathrm{ex}}
\newcommand{\corkgs}{C_{ij}^\mathrm{\,gs}(\mathbf{k})}
\newcommand{\corkfs}{C_{ij}^\mathrm{\,Fock}(\mathbf{k})}
\newcommand{\corkth}{C_{ij}^\mathrm{\,th}(\mathbf{k})}
\newcommand{\corkzmz}{\mathcal{C}_{ij}^\mathrm{ex}(\mathbf{k}_0)}
\newcommand{\corkzzmz}{\mathcal{C}_{ij}^\mathrm{ex}(k_z,0)}
\newcommand{\corqmz}{\mathcal{C}_{ij}^\mathrm{ex}(\mathbf{q})}
\newcommand{\corqzmz}{\mathcal{C}_{ij}^\mathrm{ex}(q_z)}

\newtheorem{theorem}{Theorem}
\newtheorem{corollary}{Corollary}
\newtheorem{lemma}{Lemma}
\newtheorem{definition}{Definition}
\newtheorem{proposition}{Proposition}
\newtheorem{example}{Example}
\newtheorem{remark}{Remark}

% Fourier transforms
\def\ftq#1{\mathcal{F}_{\mathbf{k}}\left[#1\right]}
\def\ftxiij#1{\mathcal{F}^{-1}_{\mathbf{r}}\left[#1\right]}
\def\ftxi#1{\mathcal{F}^{-1}_{\mathbf{r}}\left[#1\right]}
\def\ftqz#1{\mathcal{F}_{k_z}\left[#1\right]}
\def\ftxiz#1{\mathcal{F}^{-1}_{r_z}\left[#1\right]}

\newcommand{\gof}{\gamma}
\newcommand{\dof}{\Delta_{\mathrm{dd}}}

\newcommand{\bmu}{{\bm \mu}}
\newcommand{\bnu}{{\bm \nu}}
\newcommand{\bphi}{{\bm \phi}}
\newcommand{\bpsi}{{\bm \psi}}
\newcommand{\be}{{\bm e}}
\newcommand{\bff}{{\bm f}}

\def\vec#1{\bm{#1}}
\def\ket#1{|#1\rangle}
\def\bra#1{\langle#1|}
\def\braket#1#2{\langle#1|#2\rangle}
\def\ketbra#1#2{|#1\rangle\langle#2|}

\def\G#1{\Gamma_{{}_#1}}

\title{Cooperative spontaneous emission from indistinguishable atoms in arbitrary motional quantum states}

\date{\today}

\author{Fran\c{c}ois Damanet$^{1}$, Daniel Braun$^{2}$ and John Martin$^1$}

\affiliation{$^1$Institut de Physique Nucl\'eaire, Atomique et de
Spectroscopie, CESAM, Universit\'e de Li\`ege, B\^atiment B15, B - 4000
Li\`ege, Belgium}
\affiliation{$^2$Institut f\"ur theoretische Physik, Universit\"at T\"ubingen, 72076 T\"ubingen,  Germany}

\begin{abstract}
We investigate superradiance and subradiance of indistinguishable atoms with quantized
motional states, starting with an initial total state that factorizes
over the internal and external degrees of freedom of the atoms.
Due to the permutational symmetry of the motional state, the cooperative
spontaneous emission, governed by a recently derived master equation 
[F. Damanet \emph{et al.}, Phys.\ Rev.\ A \textbf{93}, 022124 (2016)], depends only on two decay rates $\gamma$ and $\gamma_0$ and
a single parameter $\Delta_{\mathrm{dd}}$ describing the dipole-dipole shifts. We solve the dynamics exactly for $N=2$ atoms, numerically for up to 30 atoms, and obtain the large-$N$-limit by a
mean-field approach. We find that there is
a critical difference $\gamma_0-\gamma$ that depends on $N$ beyond
which superradiance is 
lost. We show that exact non-trivial dark states (i.e.~states other than the ground
state with vanishing spontaneous emission)
only exist for $\gamma=\gamma_0$, and that those states (dark when $\gamma=\gamma_0$) are subradiant when $\gamma<\gamma_0$.

\end{abstract}

% 03.65.Yz : Decoherence - quantum mechanics
% 02.50.Ga : Markov processes
% 37.10.Vz, 42.50.Wk : Light - mechanical effects on atoms and molecules
% 03.75.Gg : Decoherence - Bose-Einstein condensates

\maketitle

\section{Introduction}

Cooperative spontaneous emission of light from excited atoms, which results from their common coupling with the surrounding electromagnetic field, is a central field of research in quantum optics. In a seminal paper~\cite{Dic54}, Dicke showed that spontaneous
emission can be strongly enhanced when atoms are close to each other in comparison to the wavelength of the emitted radiation. This phenomenon, called superradiance, and its counterpart corresponding to reduced spontaneous emission, called subradiance, were first considered for distinguishable atoms at fixed
positions. Depending on the geometrical arrangement of the atoms in
space, deeper analyses showed later that virtual photon exchanges
between atoms were likely to destroy
superradiance~\cite{Fri72,Fri74,Gro82}. Due to the complexity involved
by the exact treatment of these dipole-dipole interactions, analytical
characterizations of superradiance have been found only for particular geometries or for a small number of atoms \cite{Cof78,Gro82,Fre86,Fre87,Ric90,Fen13}. Yet, cooperative emission processes are not restricted to small atomic samples. They can also be observed in dilute atomic systems for which the near-field contributions of the dipole-dipole interactions are insignificant. Recent studies concern single-photon superradiance~\cite{Scu06,Scu09,Oli14, Kon16}, subradiance in cold atomic gases~\cite{Bie12, Gue16}, collective Lamb-shift \cite{Roh10, Mei14} and localization of light \cite{Akk08, Max15}. Moreover, super- and subradiance can be explained using a quantum multipath interference approach and can be simulated from the measurement of higher-order-intensity-correlation functions on atoms separated by a distance larger than the emission wavelength~\cite{Wie11,Wie15}.

The atomic motion can have a significant influence on the  spontaneous emission and scattering of light~\cite{Wil92, Jav94, You94, Ber97, Bra08, Li13}, and vice versa (see e.g.~\cite{Dal85, Min87, Coh98, Wie99,Dom03, Sch14, Sar14}). However, its role on cooperative emission processes is not yet fully
understood, especially in large laser driven atomic
systems~\cite{Zhu16}. The interplay of atomic motion and cooperative
processes has been the subject of recent
experiments~\cite{Lab06,Pel14, Jen16, Jen16b, Bro16} and can lead to
interesting effects such as supercooling of atoms \cite{Min16} or
superradiant Rayleigh scattering from a Bose-Einstein condensate (BEC) \cite{Ino99}. In
hot atomic samples, the motion can be treated classically and leads to
Doppler broadenings of the spectral lines. In (ultra)cold atomic
samples, the quantum nature of the motion and the indistinguishability
must be taken into account, as they also lead to strong modifications
of the dynamics. In this paper, we study super- and subradiance from
indistinguishable atoms, taking into account recoil, quantum
fluctuations of the atomic positions, and quantum statistics. To this end, we solve a recently derived master equation describing the cooperative spontaneous emission of light by $N$ two-level atoms in arbitrary quantum motional states~\cite{Dam16}.

Indistinguishability of atoms profoundly changes their internal
dynamics as compared to that of distinguishable atoms.  
For distinguishable atoms with classical positions, the solution of
the master equation depends considerably on the geometrical
arrangement of the atoms in space. When describing each atom as a
two-level atom, the internal state of the atoms
thus evolves in the full Hilbert space of dimension $2^N$. The same
general considerations can be made when the atomic positions are
treated quantum mechanically since despite changed rates and level
shifts the master equation \cite{Dam16}
then retains the same global form with the same Lindblad operators.
 Hence, each configuration must be dealt with case by case. However,
 for indistinguishable atoms, the global state has to be invariant under
 exchange of the atoms. For initial states that are separable
 between the internal degrees of freedom and the motional degrees of
 freedom, both internal and motional states must be invariant under permutation of atoms. Furthermore, on the time scale of
 the spontaneous emission in the optical domain, the motional state
 can be considered frozen such that the permutational symmetry of the
 motional state prevails throughout the entire emission process. This
 leads to permutationally invariant average Lindblad-Kossakowski matrix of
 emission rates and permutationally invariant dipole shifts, which
 limits the quantum dynamics also of the internal degrees of freedom to the
 permutation-invariant subspace of dimension $\mathcal{O}(N^2)$ of the global Hilbert space \cite{Cha08, Bar10}, thus greatly
 simplifying  the problem.  However, the quantum fluctuations of the
 positions of the atoms modify the cooperative effect
 of collective emission, thus leading back from superradiance to
 individual spontaneous emission for large enough quantum uncertainty
 in the positions.  

The paper is organized as follows. In section II, we discuss the
general form of the master equation for the atomic internal dynamics
derived in \cite{Dam16} in the case of indistinguishable atoms. In
particular, we show that the master equation preserves permutation
invariance of the internal state. In section III, we write the master
equation in the coupled spin basis as it is particularly suited for
permutation-invariant states. Finally, in section IV, we solve the
master equation analytically for $N=2$ and numerically for up to
$N=30$ atoms in order to study the impact of the quantization of the
atomic motion on super- and subradiance.

\section{General form of the master equation for indistinguishable atoms}

\subsection{Symmetry of the initial state}

Let us consider $N$ indistinguishable atoms in a mixture $\rho_A$. Each wave function of the mixture has to be either symmetric (bosons) or antisymmetric (fermions) under exchange of atoms. Let $P_\pi$ denote the permutation operator of a permutation $\pi$ defined through exchange of the atomic labels. We have [see Eq.~(\ref{eq:full}) in Appendix A]
\begin{equation}
P_\pi\rho_A P_{\pi'}^\dagger= (\pm 1)^{p_\pi + p_{\pi'}}\rho_A\quad \forall\,\pi, \pi'
\end{equation}
where $p_\pi$ is the parity
of the permutation $\pi$ (even 
or odd), and $(\pm 1)^{p_\pi}$ the 
phase factor picked up accordingly for bosons (upper sign) or fermions
(lower sign). Moreover, the Born approximation performed in
\cite{Dam16} assumes that the initial atomic state is separable, i.e.\
$\rho_A(0)=\rho_A^\mathrm{in}(0)\otimes \rho_A^\mathrm{ex}$ with
$\rho_A^\mathrm{ex}$ the motional density operator at time $t=0$. This implies that both internal and external states are invariant under permutation of atoms [see Eq.~(\ref{eq:bos}) in Appendix A],
\begin{equation}
\begin{aligned}
  \label{eq:bos}
  &P^\mathrm{in}_\pi\rho^\mathrm{in}_A(0) P_{\pi}^{\mathrm{in}\dagger}=\rho^\mathrm{in}_A(0)  \quad\forall\, \pi, \\[2pt]
  &P^\mathrm{ex}_\pi\rho^\mathrm{ex}_A  P_{\pi}^{\mathrm{ex}\dagger}=\rho^\mathrm{ex}_A \quad\forall\, \pi.
  \end{aligned}
\end{equation}

\subsection{Standard form}

In the interaction picture, the master equation for the reduced density matrix $\rho_A^\mathrm{in}(t)$ describing the internal dynamics of the system $A$ composed of $N$ indistinguishable atoms takes the standard form~\cite{Dam16}
\begin{equation}\label{meqsummary}
\begin{aligned}
\frac{d\rho_A^\mathrm{in}(t)}{dt} &= \mathcal{L}\left[\rho_A^\mathrm{in}(t) \right] =  - \frac{i}{\hbar}\left[ H_\mathrm{dd}, \rho_A^\mathrm{in}(t) \right] + \mathcal{D}\left[\rho_A^\mathrm{in}(t)\right],
\end{aligned}
\end{equation}
with the Liouvillian superoperator $\mathcal{L}\left[\cdot\right]$ involving the dipole-dipole Hamiltonian
\begin{equation}\label{Hamildipdip}
H_\mathrm{dd} = \hbar \Delta_{\mathrm{dd}}\, \sum_{i \neq j }^{N} \sigma_+^{(i)}\sigma_-^{(j)},
\end{equation}
with $\Delta_{\mathrm{dd}}$ the dipole-dipole shifts, and the dissipator
\begin{equation}
\label{Dissip}
\begin{aligned}
\mathcal{D}\left[\rho_A^\mathrm{in}\right] ={}& \gamma \, \sum_{i\neq j}^N \left(\sigma_-^{(j)}\rho_A^\mathrm{in}\sigma_+^{(i)}-\frac{1}{2}\left\{\sigma_+^{(i)}\sigma_-^{(j)},\rho_A^\mathrm{in}\right\}\right) \\
& +  \gamma_0\, \sum_{i=1}^N \left(\sigma_-^{(i)}\rho_A^\mathrm{in}\sigma_+^{(i)}-\frac{1}{2}\left\{\sigma_+^{(i)}\sigma_-^{(i)},\rho_A^\mathrm{in}\right\}\right),
\end{aligned}
\end{equation}
with $\gamma_0$ the single-atom spontaneous emission rate and $\gamma$ the cooperative (off-diagonal) decay rates. In Eqs.~(\ref{Hamildipdip}) and (\ref{Dissip}), $\sigma_{+}^{(j)} = (\ket{e}\bra{g})_j$ and $\sigma_{-}^{(j)} =
(\ket{g}\bra{e})_j$ are the ladder operators for atom $j$ with
$|g\rangle$ ($|e\rangle$) the lower (upper) atomic level of energy
$-\hbar \omega_0/2$ ($\hbar \omega_0/2$). Note that in
Eq.~(\ref{Hamildipdip}), we do not consider diagonal terms ($i = j$)
corresponding to the Lamb-shifts. They can be discarded by means of a
renormalization of the atomic frequency. The fact that all
off-diagonal ($i\ne j$) decay rates are \emph{equal} and all
dipole-dipole shifts are \emph{equal} for any pairs of atoms is merely a consequence of
the indistinguishability of atoms (see Appendix B for a formal
derivation).

The master equation~(\ref{meqsummary}) is valid for
\emph{arbitrary} motional quantum states and can be applied well-beyond the Lamb-Dicke regime. All effects related to the quantization of the atomic motion are
encoded in the values taken by the dipole-dipole shift $
\Delta_{\mathrm{dd}}$ and the decay rate $\gamma$. We give their exact expressions for arbitrary motional symmetric or antisymmetric states in Appendix B [see Eqs.~(\ref{indisEXCHANGE1}) and (\ref{indisEXCHANGE2})]. They depend not only on the average atomic positions (classical atomic positions)
but also on their quantum fluctuations and correlations as described by the quantum motional (external) state
$\rho_A^\mathrm{ex}$ of the atoms. In particular, their values can
strongly depend on the statistical nature (bosonic or fermionic) of
the atoms. In the next section, we give analytical expressions of $\gamma$ for BECs in different regimes. 

\subsection{Off-diagonal decay rates $\gamma$ for Bose-Einstein condensates}

We first consider the case of a non-interacting BEC confined in an isotropic harmonic trap at zero temperature. In this case, all atoms are in the same motional state $\phi(\mathbf{r}) = e^{-|\mathbf{r}|^2/4 \ell^2}/(\sqrt{2\pi}\ell)^{3/2}$ with $\ell = \sqrt{\hbar/2M\Omega}$ the width of the spatial density, $M$ the atomic mass and $\Omega$ the trap frequency. The decay rate $\gamma$ for this motional state follows from Eq.~(\ref{beEXCHANGE1}) with $\rho_1(\mathbf{r})=|\phi(\mathbf{r})|^2$ and is given by \begin{equation}\label{gammaISO}
\gamma = \gamma_0\, e^{-\eta^2}
\end{equation}
with $\eta = k_0 \ell$ the Lamb-Dicke parameter and $k_0$ the radiation wavenumber. Since the size of a BEC typically lies in the range $10-10^3\mu\mathrm{m}$ \cite{Pet08}, significant modifications of the decay rate $\gamma$ should  be observable for internal transitions in the visible and near-infrared domain.

We now consider the case of a BEC with strong repulsive interactions at zero temperature, for which the spatial density $\rho_1(\mathbf{r})$ is given in the Thomas-Fermi approximation by 
\begin{equation}\label{TFBEC}
\rho_1(\mathbf{r}) = \begin{cases} \frac{M}{4\pi\hbar^2 a}\left[\mu - V_\mathrm{ext}(\mathbf{r})\right] & \mbox{for }\mu \geqslant V_\mathrm{ext}(\mathbf{r}), \\[4pt]
0 & \mbox{for }\mu < V_\mathrm{ext}(\mathbf{r}),\end{cases}
\end{equation}
where $a$ is the scattering length, $\mu$ is the chemical potential and $V_\mathrm{ext}(\mathbf{r}) = M \Omega^2 r^2/2$ is the harmonic trap potential. Inserting Eq.~(\ref{TFBEC}) into Eq.~(\ref{beEXCHANGE1}) yields after integration
\begin{equation}\label{gammaTF}
\gamma = 225 \gamma_0 \frac{\left(3x \cos x + (x^2 - 3) \sin x\right)^2}{x^{10}}
\end{equation}
where $x = \eta \sqrt[5]{60 N a/\ell}$ with $N$ the number of atoms in the BEC. Figure~\ref{gammaThomasFermi} shows Eq.~(\ref{gammaTF})  as a function of $x$. In particular, when $x\to 0$ (small recoil), $\gamma$ tends to $\gamma_0$.

We finally study the transition from a thermal cloud to a non-interacting BEC by considering a gas of trapped bosonic atoms in thermal equilibrium at finite temperature $T$. The spatial density of the gas is given by~\cite{Lan80}
\begin{equation}\label{eq:BECfiniteT}
\rho_1(\mathbf{r}) = \frac{1}{N(2\pi \ell^2)^{\frac{3}{2}}}\sum_{k = 1}^{\infty} \frac{z^k}{\left( 1 - e^{-2 k \beta \hbar \Omega} \right)^{\frac{3}{2}}}\, e^{-\frac{r^2}{2\ell^2} \tanh\left(\frac{k\beta\hbar\Omega}{2}\right)},
\end{equation}
where $z=e^{\beta\mu}$ is the fugacity and $\beta = 1/k_BT$ with $k_B$ the Boltzmann constant. Inserting Eq.~(\ref{eq:BECfiniteT}) into Eq.~(\ref{beEXCHANGE1}) yields after integration
\begin{equation}\label{gammaISOT}
\gamma =  \frac{\gamma_0}{N^2} \Bigg[\sum_{k =1}^{\infty} \frac{z^{k} e^{3 k \beta \hbar \Omega}}{(1- e^{k \beta \hbar \Omega})^3}\, e^{-\frac{\eta^2}{2}\coth\left(\frac{k\beta\hbar\Omega}{2}\right) }\Bigg]^2.
\end{equation}
For $z\to 0$, Eq.~(\ref{eq:BECfiniteT}) tends to a thermal cloud profile $\rho_1(\mathbf{r})=e^{-(r/2R)^2}/(2\pi R^2)^{3/2}$ with $R=\sqrt{k_BT/m\Omega^2}$ and we get
\begin{equation}\label{gammathermal}
\gamma = \gamma_0\, e^{-k_0^2R^2}
\end{equation}
where $e^{-k_0^2R^2}$ is the Debye-Waller factor. Figure~\ref{gammaBEC} shows Eq.~(\ref{gammaISOT}) as a function of $\eta$ for $\beta \hbar \Omega = 1/10$ and different values of the fugacity. The curves $\gamma(\eta)$ switch gradually from Eq.~(\ref{gammathermal}) for $z=0$ to Eq.~(\ref{gammaISO}) for $z=1$ as the fugacity is increased, as a consequence of the formation of a condensed phase (see the inset of Fig.~\ref{gammaBEC}). For fixed Lamb-Dicke parameter and $\beta\hbar\Omega$, $\gamma$ increases monotonically with the fugacity. Hence, cooperative effects will be more pronounced when all atoms are in the condensed phase (pure BEC).

\begin{figure}
\begin{center}
\includegraphics[width=0.46\textwidth]{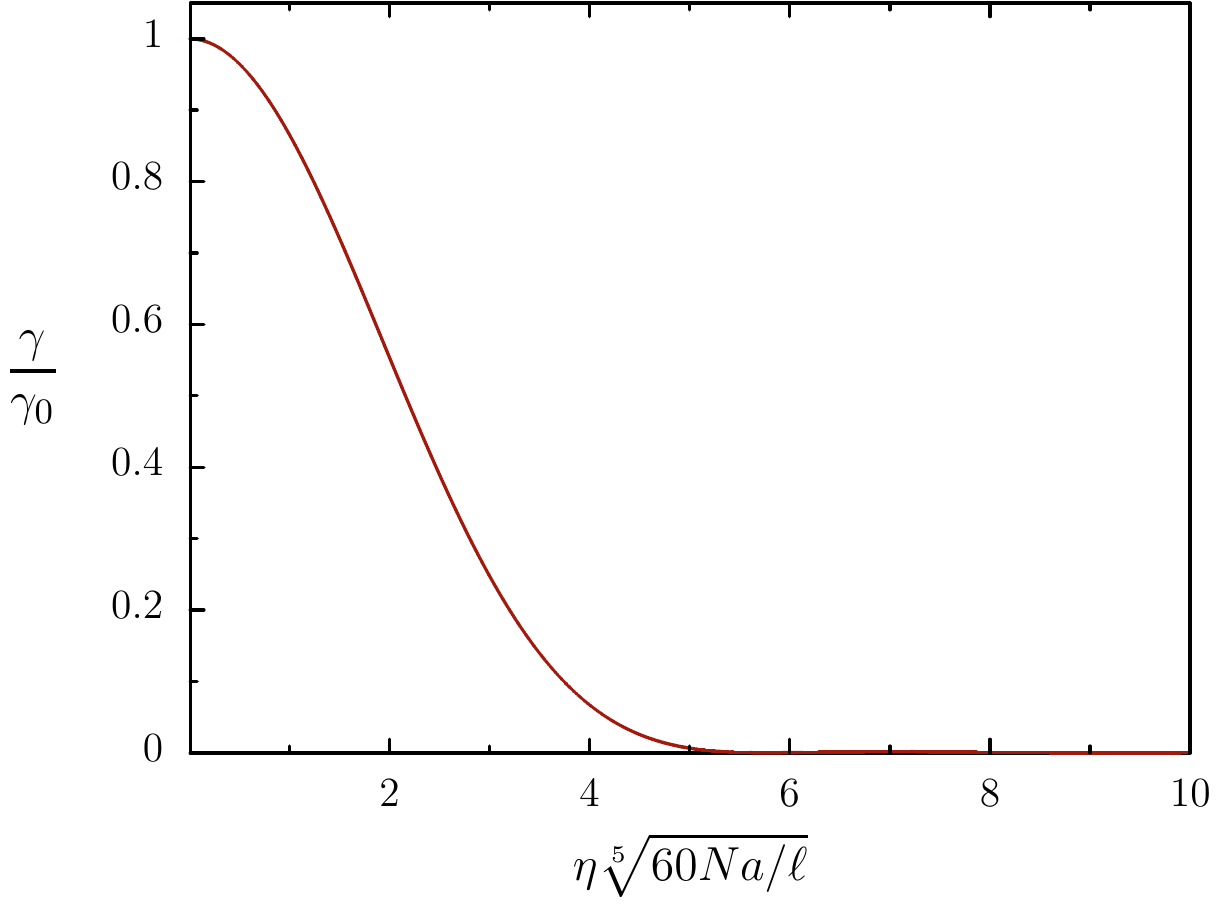}
\end{center}
\caption{Off-diagonal decay rate $\gamma$ as a function of the dimensionless variable $x = \eta \sqrt[5]{60 N a/\ell}$ for a BEC with strong repulsive interactions in the Thomas-Fermi limit.}\label{gammaThomasFermi}
\end{figure}

\begin{figure}
\begin{center}
\includegraphics[width=0.475\textwidth]{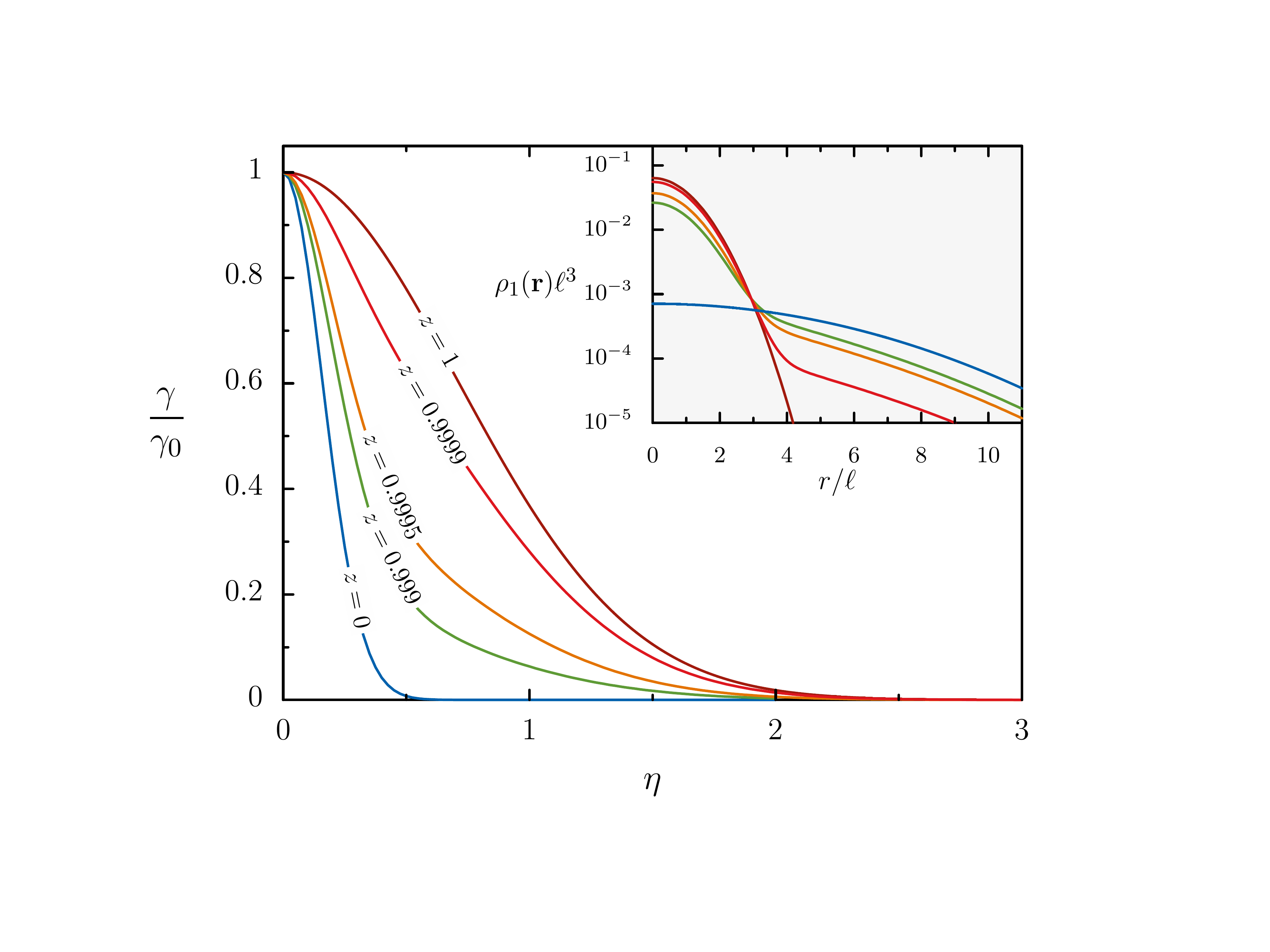}
\end{center}
\caption{Off-diagonal decay rate $\gamma$ as a function of the Lamb-Dicke parameter $\eta$ for a gas of trapped bosonic atoms at thermal equilibrium for $\beta \hbar \Omega = 1/10$ and different values of the fugacity, from $z=0$ (left) to $z=1$ (right). The inset shows spatial density profiles for the same parameters.}\label{gammaBEC}
\end{figure}

\subsection{Lindblad form}

Before we discuss the Lindblad form of the master equation (\ref{meqsummary}), let us note that the dipole-dipole Hamiltonian (\ref{Hamildipdip}) can be rewritten in terms of collective spin operators only as
\begin{equation}\label{Hamildipdip2}
H_\mathrm{dd} = \hbar \Delta_{\mathrm{dd}} \Big[ J_+ J_- - \frac{1}{2} \left( N \mathbb{1} + 2 J_z \right) \Big],
\end{equation}
where $\mathbb{1}$ is the identity operator acting on the internal atomic states, $J_\pm = \sum_{j = 1}^N \sigma_\pm^{(j)}$ are the collective spin ladder operators, and $J_z = \frac{1}{2}\sum_{j = 1}^N \sigma_z^{(j)}$ with $\sigma_z^{(j)} = (\ket{e}\bra{e} - \ket{g}\bra{g})_j$. As for the dissipator (\ref{Dissip}), it also involves individual spin operators and can be rewritten as
\begin{multline}
\label{Dissip2}
\mathcal{D}\left[\rho_A^\mathrm{in}\right] = \gamma \left( J_- \rho_A^\mathrm{in} J_+ - \frac{1}{2} \Big\{ J_+ J_-, \rho_A^\mathrm{in} \Big\} \right) \\
+ (\gamma_0-\gamma) \left( \sum_{i=1}^N \sigma_-^{(i)}\rho_A^\mathrm{in} \sigma_+^{(i)}-\frac{1}{4}\Big\{N \mathbb{1} + 2 J_z,\rho_A^\mathrm{in}\Big\}\right).
\end{multline}

The Lindblad form is obtained from the diagonalization of the $N \times N$ matrix of decay rates
\begin{equation}\label{gammaMatrix0} 
\matgamma= \begin{pmatrix}
\gamma_0& \gof & \ldots  & \gof \\
\gof & \gamma_0 & \ldots  & \gof \\
\vdots & \vdots & \ddots  & \vdots \\
\gof & \gof & \ldots & \gamma_0 \\
\end{pmatrix},
\end{equation}
with $|\gof| \leqslant \gamma_0$~\cite{Dam16}. Associated with each eigenvector with non-zero eigenvalue $\Gamma_\ell$ is a Lindblad operator $F_\ell$. Degenerate eigenvalues give rise to several Lindblad operators. The matrix (\ref{gammaMatrix0}) has eigenvalues 
\begin{subequations}
\begin{align}
& \Gamma_1 = \gamma_0 + (N-1)\gof\equiv N\gamma+\Delta\gamma,\\[3pt]
& \Gamma_2 = \gamma_0 - \gof \equiv \gb,
\end{align}
\end{subequations}
with onefold and $(N-1)$-fold degeneracy, respectively. For the dynamics
to be Markovian, the matrix $\matgamma$ has to be positive. This
implies that $-\gamma_0/(N-1)\leqslant\gamma\leqslant \gamma_0$ and
$0\leqslant \gb\leqslant \gamma_0\,N/(N-1)$. An eigenvector
$\mathbf{v}_1$ with the largest eigenvalue ($\Gamma_1$) is the vector
with all components equal to $1/\sqrt{N}$. The corresponding Lindblad
operator is $F_1=J_-/\sqrt{N}$. The remaining eigenvectors with
degenerate eigenvalue $\gb$ span the subspace $\mathbb{C}^{N-1}$
orthogonal to $\mathbf{v}_1$ and lead to Lindblad operators
$F_\ell$.
The Lindblad form of the dissipator thus reads 
\begin{multline}
\label{DissipLindblad}
\mathcal{D}\left[\rho_A^\mathrm{in}\right] = \frac{\Gamma_1}{N} \left( J_- \rho_A^\mathrm{in}  J_+ - \frac{1}{2} \Big\{ J_+  J_-, \rho_A^\mathrm{in} \Big\} \right) \\
+ \gb \left( \sum_{\ell=2}^N F_\ell\rho_A^\mathrm{in} F_\ell^\dagger-\frac{1}{2}\Big\{F_\ell^\dagger F_\ell,\rho_A^\mathrm{in}\Big\}\right). 
\end{multline}

When $\gb = 0$ ($\gamma = \gamma_0$), the system evolves under the sole action of the collective spin operator $J_-$. As a consequence, starting from an internal symmetric state, the dynamics is restricted to the symmetric subspace of dimension $N+1$. This is the superradiant regime \cite{Gro82}.

When $\gb > 0$ ($\gamma < \gamma_0$), all the additional Lindblad operators are involved in the dynamics. The superoperator multiplying $\gb$ in Eq.~(\ref{DissipLindblad}) can be rewritten as
\begin{equation}
\sum_{i=1}^N \sigma_-^{(i)} \,\boldsymbol{\cdot}\, \sigma_+^{(i)}  - \frac{J_- \boldsymbol{\cdot}\, J_+}{N}
-\frac{1}{2} \left\{ \frac{N \mathbb{1}}{2} + J_z - \frac{J_+J_-}{N}, \boldsymbol{\cdot} \right\}.
\end{equation}
Hence, it cannot be expressed as a function of collective spin operators only. However, it affects each atom identically. Therefore, the Liouvillian superoperator does not distinguish between atoms and commutes with $P_\pi$ for all permutations $\pi$, i.e.\
\begin{equation}
P_\pi \mathcal{L}[\rho] P^{\dagger}_\pi = \mathcal{L}[P_\pi \rho P^{\dagger}_\pi] \quad \forall \, \rho, \forall\,\pi.
\end{equation}
It couples symmetric states with the broader class of permutation-invariant states. These states, denoted hereafter by $\rho_\PI$, are states satisfying~\cite{Nov13}
\begin{equation}\label{pidm}
\rho_\PI = P_\pi \rho_\PI P_\pi^{\dagger} \hspace{1cm} \forall \, \pi.
\end{equation}
They act on a subspace whose dimension grows only as $N^2$~\cite{Cha08, Bar10}. 

\section{Master equation in the coupled spin basis}

From now on, we will denote the internal density matrix $\rho_A^\mathrm{in}$ by $\rho$. In this section, we express the master equation (\ref{meqsummary}) in the coupled spin basis, which is particularly suited for the study of permutation-invariant states.

\subsection{Coupled spin basis}
The Hilbert space $\mathcal{H}$ of an ensemble of $N$ two-level systems admits the Wedderburn decomposition \cite{Dar05, Bus06, Bac06, Nov13}
\begin{equation}\label{Hdecomp}
\mathcal{H} = (\mathbb{C}^2)^{\otimes N} \simeq \bigoplus_{J = J_\mathrm{min}}^{N/2} \mathcal{H}_J \otimes \mathcal{K}_J,
\end{equation}
with $J_\mathrm{min} = 0$ for even $N$ and $1/2$ for odd $N$. In Eq.~(\ref{Hdecomp}), $\mathcal{H}_J$ is the \textit{representation} space of dimension $2J+1$ on which the irreducible representations of the group SU(2) act.
The number of degenerate irreducible representations with total angular momentum $J$ is equal to the dimension
\begin{equation}
d_N^J = \frac{(2J+1)N!}{(N/2 - J)!(N/2+J+1)!}
\end{equation}
of the \textit{multiplicity} space $\mathcal{K}_J$ on which the irreducible representations of the symmetric group $S_{N}$ act.
The total Hilbert space $\mathcal{H}$ is therefore spanned by the states $\ket{J,M,\kJ}  \equiv \ket{J,M} \otimes \ket{\kJ}$, where $\ket{J, M}$ are basis states of the subspaces $\mathcal{H}_J$ ($J = J_\mathrm{min}, \dotsc, N/2$; $M = -J, \dotsc, J$) and $\ket{\kJ}$ are basis states of the subspaces $\mathcal{K}_J$ ($\kJ = 1, \dotsc, d_N^J$). The $2^N$ basis states $\left\{\ket{J,M,\kJ}\right\}$ form the coupled spin basis \cite{Bar10}. By construction, $\ket{J,M,\kJ}$ are spin-$J$ states satisfying
\begin{equation}\label{eigensystem1}
\begin{aligned}
& \mathbf{J}^2 \ket{J,M,\kJ}= J(J+1) \ket{J,M, \kJ}, \\
& J_z \ket{J,M,\kJ}= M \ket{J,M, \kJ}, \\ 
& J_\pm \ket{J,M,\kJ} = \sqrt{(J\mp M)(J\pm M + 1)}\ket{J,M\pm 1, \kJ}
\end{aligned}
\end{equation}
with $\mathbf{J}^2 = J_x^2 + J_y^2 + J_z^2$ and $J_m =
\frac{1}{2}\sum_{j = 1}^N \sigma_m^{(j)}$ ($m = x,y,z$). The
degenerate structure of the decomposition (\ref{Hdecomp}) is depicted
in the Bratteli diagram shown in Fig.~\ref{bratteli}. There are
$d_N^J$ ways to obtain an angular momentum $J$ from the coupling of
$N$ spins $1/2$, each way being associated with a path in the Bratteli
diagram. The quantum number $\kJ = 1,\dotsc, d_N^J$ enables one to distinguish these different paths.

\begin{figure}
\begin{center}
\includegraphics[width=0.475\textwidth]{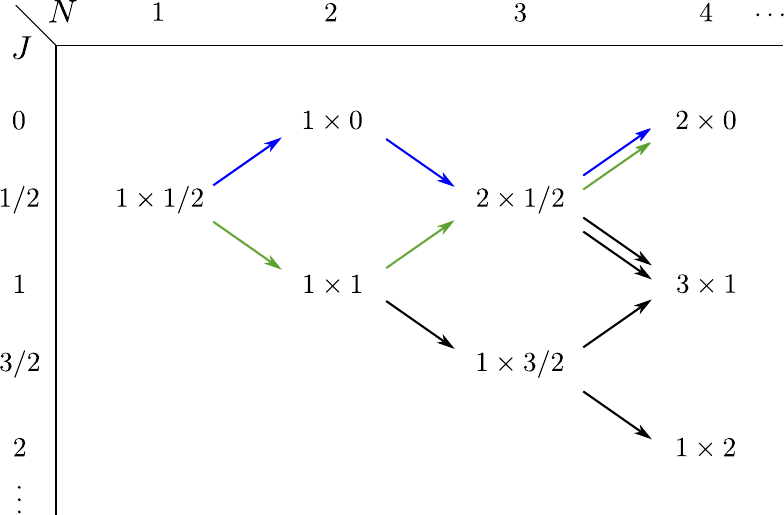}
\end{center}
\caption{Bratteli diagram representing the degeneracy structure $d_n^J \times J$ of $N$ coupled spins $1/2$. The two colored paths leading to the same angular momentum $J = 0$ correspond to two different values of the quantum number $k_{J=0}$.}\label{bratteli}
\end{figure}

\subsection{Permutation-invariant states in the coupled spin basis}

According to the Schur-Weyl duality \cite{Arn16,Chr06}, any permutation $P_\pi$ acts only on the multiplicity subspaces $\mathcal{K}_J$ [see decomposition (\ref{Hdecomp})] and thus has the form
\begin{equation}
P_\pi = \bigoplus_{J = J_\mathrm{min}}^{N/2} \mathbb{1}_{\mathcal{H}_J} \otimes P_J(\pi),
\end{equation}
where $\mathbb{1}_{\mathcal{H}_J}$ is the identity operator on $\mathcal{H}_J$ and $P_J(\pi)$ is an irreducible representation of $S_N$ of dimension $d_N^J$. A permutation-invariant mixed state $\rho_\PI$ commutes with $P_\pi$ for any permutation $\pi$ [see Eq.~(\ref{pidm})] and thus admits in the coupled spin basis a block-diagonal form~\cite{Nov13},
\begin{equation}\label{PI}
\begin{aligned}
\rho_\PI &=  \bigoplus_{J = J_\mathrm{min}}^{N/2} \rho_J \otimes \mathbb{1}_{\mathcal{K}_J}, \\
\end{aligned}
\end{equation}
where $\mathbb{1}_{\mathcal{K}_J}$ is the identity operator on $\mathcal{K}_J$ and
\begin{equation}
\begin{aligned}
\rho_J &= \sum_{M,M' = -J}^J  \rho_{J}^{M,M'} \, \ket{J, M } \bra{J, M'},
\end{aligned}
\end{equation}
with the density matrix elements
\begin{equation}\label{dmelem}
\rho_{J}^{M,M'} \equiv  \langle J,M,\kJ | \rho_\PI| J,M',\kJ \rangle \quad \forall \, \kJ.
\end{equation}
The block-diagonal form illustrated in Fig.~\ref{blockdiagonal} shows that $\rho_\PI$ does not contain any coherences between blocks of different angular momenta $J$. For each $J$, there are $d_N^J$ identical subblocks. Since the matrix elements in these blocks do not depend on the label $\kJ = 1, \dotsc, d_N^J$, the number of real parameters needed to specify a permutation-invariant state $\rho_\PI$ corresponds to the sum of the density matrix elements of all $\rho_J$
\begin{equation}\label{dparameters}
\sum_{J = J_\mathrm{min}}^{N/2} (2J+1)^2 = \frac{1}{6}(N+1)(N+2)(N+3)=\mathcal{O}(N^3).
\end{equation}
This number is much smaller than the total number ($2^{2N}-1$) of a general $N$-atom density operator and highlights the
convenience of this representation. 
\begin{figure}
\begin{center}
\includegraphics[width=0.475\textwidth]{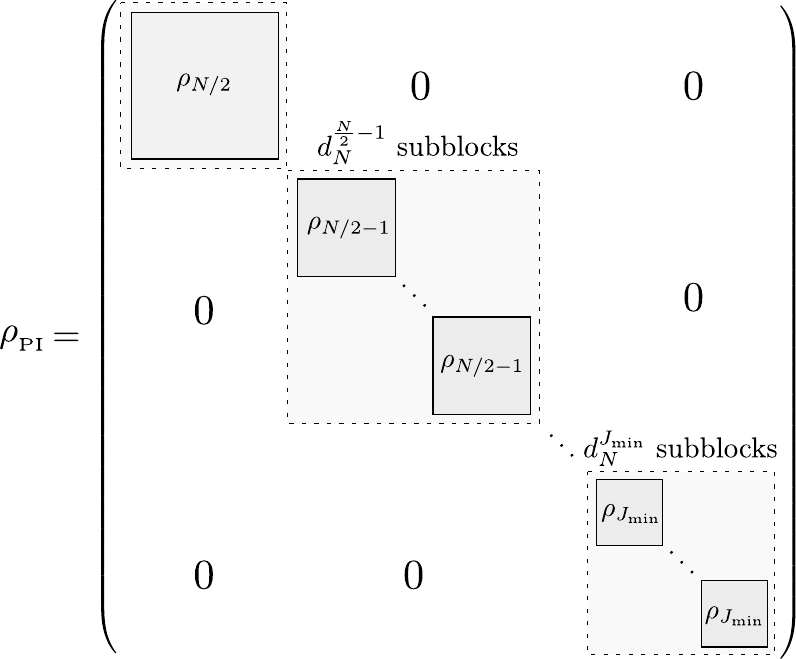}
\end{center}
\caption{Block-diagonal form of the density matrix representing a
  permutation-invariant state in the coupled spin basis. To each value of the angular momentum $J$ corresponds $d_N^J$ subblocks of dimension $(2J+1)\times(2J+1)$. The block with $J = N/2$ is unique and is spanned by symmetric states.}\label{blockdiagonal}
\end{figure}

Note that a symmetric mixed state $\rho_S$ is just a particular case of permutation-invariant state (\ref{PI}) with $ \rho_{J}^{M,M'} = 0$ for $J \neq N/2$. The non-vanishing matrix elements of symmetric states all lie in the upper block $\rho_{N/2}$ of dimension $N+1$ depicted in Fig.~\ref{blockdiagonal}. 

\subsection{Projection of the master equation in the coupled spin basis}
\label{mepiDipDip}

Let us write the master equation (\ref{meqsummary}) in terms of matrix
elements of the density operator in the coupled spin basis. By
inserting Eq.~(\ref{PI}) into Eq.~(\ref{meqsummary}) and upon using
Eqs.~\eqref{Hamildipdip2}
and (\ref{eigensystem1}), we get
\begin{widetext}
\begin{equation}\label{conspart}
\left[H_\mathrm{dd}, \rho(t)\right]
= \hbar \Delta_{\mathrm{dd}}   \sum_{J = J_\mathrm{min}}^{N/2} \sum_{M,M' = -J}^{J} \rho_{J}^{M,M'}(t) \big(M'^2-M^2\big)\, \ket{J, M} \bra{J, M'} \otimes \mathbb{1}_{\mathcal{K}_J},
\end{equation}

\begin{equation}\label{Dissip4}
\begin{aligned}
\mathcal{D}\left[\rho(t)\right] = \sum_{J = J_\mathrm{min}}^{N/2} \sum_{M,M' = -J}^{J} \rho_{J}^{M,M'}(t)   \Bigg[ & \, \gamma \, A_-^{J,M} A_-^{J,M'}\, \ket{J, M-1} \bra{J, M'-1}\otimes \mathbb{1}_{\mathcal{K}_J} \Bigg.  \\ 
&\,  - \frac{1}{2}\left( \gamma\big(A_-^{J,M}\big)^{2}+\gamma\big(A_-^{J,M'}\big)^{2}+ \gb (N+M+M') \right) \ket{J, M} \bra{J, M'} \otimes \mathbb{1}_{\mathcal{K}_J}\\
&\,  +  \gb \sum_{j=1}^N\sigma_-^{(j)} \Big[\ket{J, M} \bra{J, M'}\otimes \mathbb{1}_{\mathcal{K}_J}\Big] \sigma_+^{(j)} \,\Bigg]. 
\end{aligned}
\end{equation}

The last term in Eq.~(\ref{Dissip4}) cannot be written solely in terms of collective spin operators but affects each atom identically. It has been evaluated in~\cite{Cha08} and reads
\begin{equation}
\begin{aligned}\label{decohterm}
\sum_{j=1}^N\sigma_-^{(j)} \Big[\ket{J, M} \bra{J, M'}\otimes \mathbb{1}_{\mathcal{K}_J}\Big] \,\sigma_+^{(j)} = & \, \frac{1}{2J}A_-^{J,M}A_-^{J,M'}\left( 1 + \frac{\alpha_N^{J+1}(2J+1)}{d_{N}^J (J+1)} \right) \,\ket{J, M-1} \bra{J, M'-1}\otimes \mathbb{1}_{\mathcal{K}_J} \\
 & \, +  \frac{B_-^{J,M}B_-^{J,M'} \alpha_N^J}{2 J d_N^{J-1}} \,\ket{J-1, M-1} \bra{J-1, M'-1}\otimes \mathbb{1}_{\mathcal{K}_{J-1}} \\ 
 & \, +  \frac{D_-^{J,M}D_-^{J,M'} \alpha_N^{J+1}}{2 (J +1) d_N^{J+1}}  \,\ket{J+1, M-1} \bra{J+1, M'-1}\otimes \mathbb{1}_{\mathcal{K}_{J+1}},
\end{aligned}
\end{equation}
\end{widetext}
where
\begin{align}
& A_\pm^{J,M} = \sqrt{(J\mp M)(J\pm M+1)},\label{defCoeff1}\\
& B_-^{J,M} = -\sqrt{(J+M)(J+M-1)} ,
\label{defCoeff2}\\
& D_-^{J,M} = \sqrt{(J-M+1)(J-M+2)} , \label{defCoeff3}
\end{align}
and 
\begin{equation}\label{defCoeff4}
\alpha_{N}^{J} = \sum_{J' = J}^{N/2} d_{N}^{J'}.
\end{equation}

Equation~(\ref{conspart}) shows that dipole-dipole interactions do not couple blocks of different angular momentum $J$, but couple non-diagonal ($M \neq M'$) density matrix elements within a block. The term~(\ref{decohterm}) describes transitions giving rise to energy loss due to photon emissions, since it reduces the value of the quantum numbers $M$ and $M'$ by one unit. Such transitions from a block of angular momentum $J$ occur either within a same block or in neighboring blocks of angular momentum $J \pm 1$. The former preserve the symmetry of the state while the latter modify it.

By injecting Eqs.~(\ref{conspart}) and (\ref{Dissip4}) into the master equation~(\ref{meqsummary}) and projecting onto the  states $\ket{J, M, \kJ}$, we get a system of  $\mathcal{O}(N^3)$ [see Eq.~(\ref{dparameters})] differential equations for the density matrix elements $\rho_{J}^{MM'}(t)$ that reads
\begin{widetext}
\begin{equation}\label{diffeq}
\frac{d\rho_{J}^{M,M'}(t)}{dt} = - \G{J^{M,M'}}^{(1)}  \,\rho_{J}^{M,M'}(t) 
+ \G{J^{M+1,M'+1}}^{(2)} \, \rho_{J}^{M+1,M'+1}(t) 
+ \G{{J+1}^{M+1,M'+1}}^{(3)} \,\rho_{J+1}^{M+1,M'+1}(t)
 + \G{{J-1}^{M+1,M'+1}}^{(4)} \,\rho_{J-1}^{M+1,M'+1}(t),
\end{equation}
with
\begin{equation}
\begin{aligned}\label{diffeqb}
& \G{J^{M,M'}}^{(1)} =  i \Delta_{\mathrm{dd}} 
(M'^2-M^2) +  \frac{\gamma}{2} \Big[\big(A_-^{J,M}\big)^{2}+\big(A_-^{J,M'}\big)^{2}\Big] +   \frac{\gb}{2}(N+M+M'),  \\
& \G{J^{M+1,M'+1}}^{(2)} = A_{+}^{J,M} A_{+}^{J,M'} \left[ \gamma \,  +  \frac{ \gb}{2J}\left( 1 + \frac{\alpha_N^{J+1}(2J+1)}{d_{N}^J (J+1)} \right) \right],\\
& \G{{J+1}^{M+1,M'+1}}^{(3)} = \gb \, \frac{B_-^{J+1,M+1}B_-^{J+1,M'+1} \alpha_N^{J+1}}{2 (J+1) d_N^{J}},\\
& \G{{J-1}^{M+1,M'+1}}^{(4)} =  \gb \, \frac{D_-^{J-1,M+1}D_-^{J-1,M'+1} \alpha_N^{J}}{2 J  d_N^{J}}.
\end{aligned} 
\end{equation}
\end{widetext}
Equations~(\ref{diffeqb}) for the transition rates show that the
populations $\rho_{J}^{M,M}$ are decoupled from the coherences
$\rho_{J}^{M,M'}$ ($M \neq M'$). More specifically, coherences specified by $M,M'$ 
are only coupled to coherences with the same difference $M-M'$, and populations
$\rho_{J}^{M,M}$ can only feed populations
$\rho_{J'}^{M',M'}$ with $M' \leqslant M$ and $J' \geqslant
(J-M)/2$. This can be seen from Eqs.~(\ref{diffeq}) and (\ref{diffeqb}) and Fig.~\ref{rates}, which shows
the couplings between the populations together with the corresponding rates. Indeed, in Eq.~(\ref{diffeq}), the derivative of $\rho_{J'}^{M',M'}$ depends only on density matrix elements with equal or larger quantum numbers $M$, which implies that starting from a state with a given $M$, only states with $M' \leqslant M$ can be populated during the dynamics. As for the quantum number $J'$, it can decrease or increase through the channels with rates $\Gamma^{(3)}$ and $\Gamma^{(4)}$ (see Fig.~\ref{rates}). However, it cannot decrease indefinitely. Consider the initial state $\ket{J,M}$: All states $\ket{J-Q, M-Q}$ with positive half-integer $Q$ can be populated provided that $J-Q\geqslant J_\mathrm{min}$ and $J-Q\geqslant M-Q\geqslant -(J-Q)$. The first inequality of the latter expression is always satisfied since $M \leqslant J$, but the second inequality imposes $Q \leqslant (J+M)/2$. This in turn implies the minimal value $(J-M)/2$ for the quantum number $J' \equiv J-Q$.

\begin{figure}
\begin{center}
\includegraphics[width=0.475\textwidth]{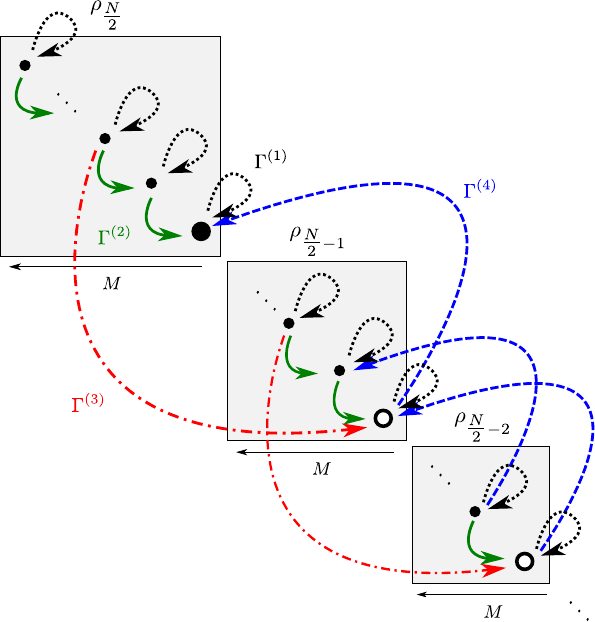}
\end{center}
\caption{Couplings between the populations $\rho_{J}^{M,M}$ (small closed circles) lying in the different blocks $\rho_J$ of angular momentum $J=N/2,N/2-1,N/2-2,\dotsc$ (gray squares), as described by Eq.~(\ref{diffeq}). The large closed and open circles at the bottom right of each block are the populations $\rho_{J}^{-J,-J}$ corresponding to subradiant states when $\gb = 0$ (see Sec.~\ref{secSubradiance}). The arrows show the different couplings between populations characterized by the rates $\Gamma^{(r)}$ (with $r = 1, 2,3,4$ and where the subscripts have been dropped for the sake of clarity). The rates $\Gamma^{(1)}$ and $\Gamma^{(2)}$ are related to transitions within a block while the rates $\Gamma^{(3)}$ and $\Gamma^{(4)}$ (proportional to $\gb$) are related to transitions between different blocks $\rho_J$. This diagram shows that starting with the initial condition $\rho_{J}^{M,M}(0) = 1$, only populations $\rho_{J'}^{M',M'}$ with $M' \leqslant M$ and $J' \geqslant  (J-M)/2$ can be non-zero during the radiative decay. When $\Delta\gamma > 0$, $\Gamma^{(3)}$ and $\Gamma^{(4)}$ are non-zero and the state $|N/2,-N/2\rangle$ (large closed circles) is the only stationary state for any initial conditions.}
\label{rates}
\end{figure}

\section{Solutions of the master equation}

The solutions of the master equation for indinstinguishable atoms only
involve the rates $\gamma$, $\gb = \gamma_0 - \gamma$, and
$\Delta_\mathrm{dd}$. In this section, we compute numerical solutions
up to $30$ atoms for different values of these rates. The solutions
allow us to study the modifications of super- and subradiance arising
from a proper quantum treatment of the atomic motion. In addition, we obtain
analytical results for large $N$ by applying a mean-field
approximation. 

In order to quantify the modifications in the release of energy from the atomic system, we calculate the normalized radiated energy rate~\cite{Gro82}
\begin{equation}\label{RER}
I(t) = -\frac{d}{dt}\langle J_z \rangle (t).
\end{equation}
For permutation-invariant states (\ref{PI}), Eq.~(\ref{RER}) can be expressed in terms of the populations $\rho_{J}^{M,M}$ as
\begin{equation}\label{RER2}
I(t) = -\sum_{J = J_\mathrm{min}}^{N/2} d_{N}^{J} \sum_{M = -J}^{J} M \frac{d \rho_{J}^{M,M}(t)}{dt}.
\end{equation}
By inserting Eq.~(\ref{diffeq}) into (\ref{RER2}) and after algebraic manipulations, we get
\begin{equation}\label{RER3}
I(t) = \sum_{J = J_\mathrm{min}}^{N/2} d_{N}^{J} \sum_{M = -J}^{J} c_{J}^{M} \, \rho_{J}^{M,M}(t)
\end{equation}
with positive coefficients $c_{J}^{M}$ given by
\begin{multline}\label{RER3coeff}
c_{J}^{M} = \big(J+M\big)\big(J-M+1\big) \, \gamma + \left(M + \frac{N}{2}\right) \gb.
\end{multline}

\subsection{Superradiance}

The superradiance phenomenon is usually observed when the atoms are initially in a symmetric internal state $\ket{N/2, M}$. In this section, we choose for initial state the symmetric state $\ket{N/2,N/2}\equiv |e,e,\dotsc,e\rangle$. This choice allows us to study the superradiant radiative cascade starting from the highest energy level.

\subsubsection{Analytical results for $2$ atoms}

For two atoms, a simple analytical solution of the master equation can be obtained and is given in Appendix C. 
For the initial condition $\rho(0) = \ket{1,1}\bra{1,1}\equiv \ket{e,e}\bra{e,e}$, the radiated energy rate (\ref{RER3}) resulting from the solution (\ref{fullsolutionN2}) given in the Appendix reads
\begin{equation}\label{eq:it}
\begin{aligned}
I(t) = \frac{e^{-2(\gof + \gb)t}} {(2 \gof + \gb) \gb} \bigg[ & \, (2 \gof + \gb)^2 \gb  + \gb^2 (2 \gof + \gb) \\ 
& \,+ (2 \gof + \gb)^3 \Big(e^{\gb t} - 1\Big) \\
& \, + \gb^3 \Big(e^{(2 \gof + \gb) t} - 1\Big) \bigg].
\end{aligned}
\end{equation}
In the absence of quantum fluctuations of the atomic positions and for
colocated atoms \cite{Dam16}, i.e.\ when $\gb=0$ ($\gamma =
\gamma_0$), pure superradiance occurs during which all symmetric Dicke states $\ket{1,1},\ket{1,0}$ and $\ket{1,-1}$ are gradually populated. In this case, Eq.~(\ref{eq:it}) reduces to the superradiant radiated energy rate
\begin{equation}\label{eq:itsup}
I(t) = 2 \gamma_0\,  e^{-2 \gamma_0 t} (1 + 2 \gamma_0 t).
\end{equation}
When $\gb > 0$, the singlet state $\ket{0,0}$ is coupled to the symmetric Dicke states and the radiated energy rate is reduced at small times as can be seen in Fig.~\ref{Intensity}.

When $\gamma=0$, $\gb=\gamma_0$ and Eq.~(\ref{eq:it}) reduces to the pure exponential decay characteristic of individual spontaneous emission 
\begin{equation}
I(t) = 2 \gamma_0\, e^{- \gamma_0 t}. \label{eq:itspo}
\end{equation}

\begin{figure}
\begin{center}
\includegraphics[width=0.47\textwidth]{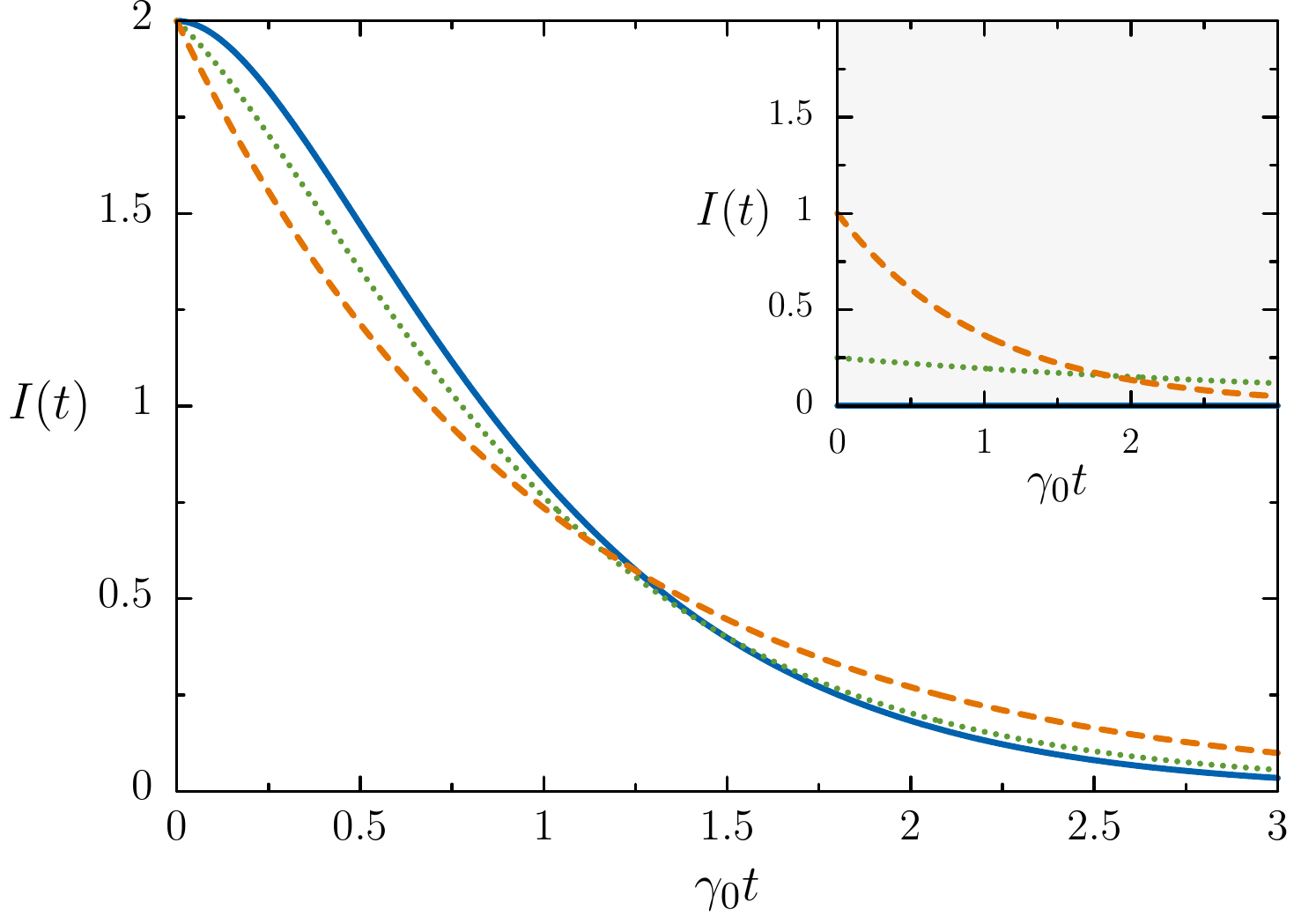}
\end{center}
\caption{Radiated energy rate for two atoms in the initial state $\ket{ee}$ as a function of time for $\gof=\gamma_0$ (blue solid curve), $\gof=3\gamma_0/4$ (green dotted
  curve), $\gof=0$ (orange dashed curve). The blue solid curve
  corresponds to pure superradiance [Eq.~(\ref{eq:itsup})], the orange
  dashed curve to independent spontaneous emission
  [Eq.~(\ref{eq:itspo})] and the green dotted curve to altered
  superradiance [Eq.~(\ref{eq:it})]. The inset shows the radiated energy rate for the initial state $\ket{0,0}$ and the same parameters.} \label{Intensity} 
\end{figure}

\subsubsection{Numerical results for $N > 2$}

In this section, we solve numerically the set of coupled equations (\ref{diffeq}) for the initial condition $\rho(0) = \rho_{N/2}^{N/2,N/2} = |e,e,\dotsc,e\rangle\langle e,e,\dotsc,e|$ and for different values of $\gb$. We then compute the radiated energy rate (\ref{RER3}). Figure~\ref{RItime} shows $I(t)$ as a function of time from $3$ to $30$ atoms, where each panel corresponds to a different value of $\gb$. For $\gb=0$, pure superradiance occurs (first panel). For $\gb=\gamma_0$, the radiated energy rate decreases according to $I(t) = N \gamma_0\, e^{-\gamma_0 t}$, as is typical of individual spontaneous emission (last panel). The middle panels show the crossover between these two regimes. Figure \ref{RItime30} is a three-dimensional plot of $I(t)$ showing the crossover for $N=30$. 

\begin{figure*}[hbt]
\begin{center}
\includegraphics[width=0.95\textwidth]{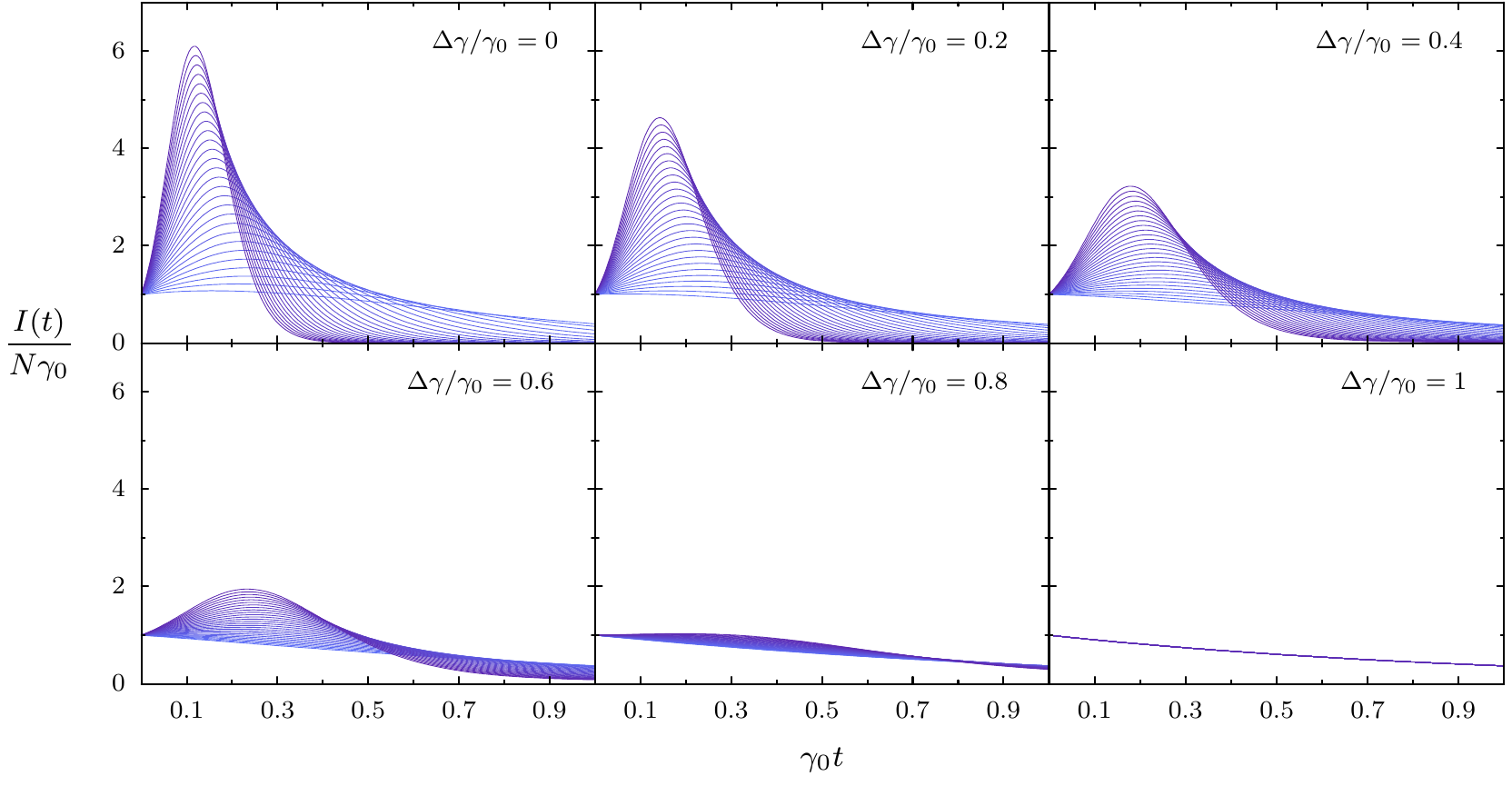}
\end{center}
\caption{Radiated energy rate as a function of time for different values of $\gb = \gamma_0 - \gamma$ (corresponding to the different panels) and different number of atoms ($N=3,\dotsc,30$ from bottom to top on the left of each panel). The case $\gb = 0$ (pure superradiance) is illustrated in the first panel while the case $\gb = \gamma_0$ corresponding to independent spontaneous emissions is illustrated in the last panel.}\label{RItime}
\end{figure*}

\begin{figure}[hbt]
\begin{center}
\includegraphics[width=0.45\textwidth]{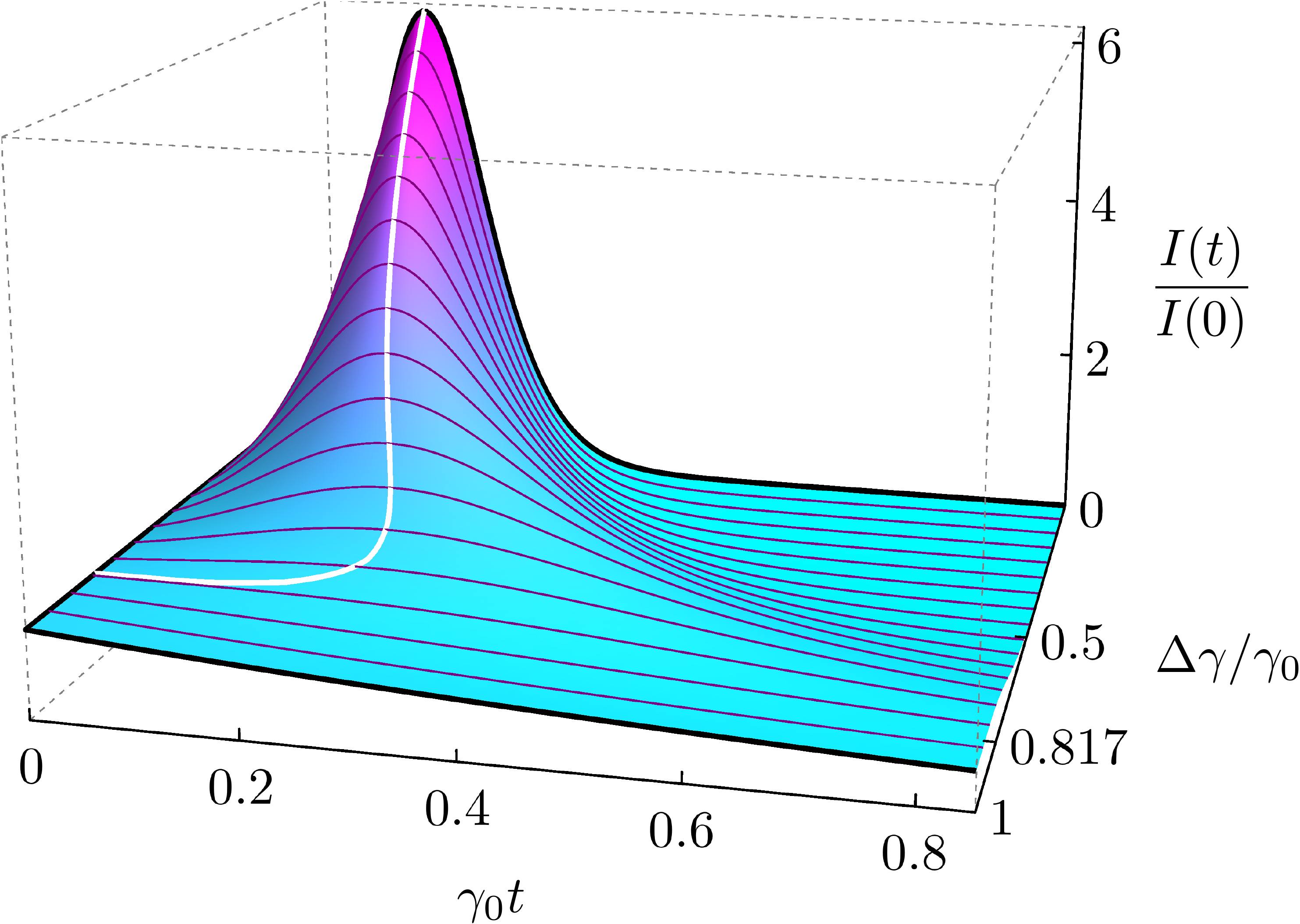}
\end{center}
\caption{Radiated energy rate as a function of time and $\gb$ for $N = 30$ atoms. The superradiant pulse progressively disappears as $\gb$ increases from $0$ to $\gb^*=0.817\,\gamma_0$. For $\gb = \gamma_0$, $I(t)$ decays exponentially at a rate $\gamma_0$. The white line indicates the location of the maximum of the pulse.}\label{RItime30}
\end{figure}

In order to characterize the superradiant pulse in the intermediate regime, we compute its relative height $A_I$ and the time $t_I$ at which its maximum occurs. These quantities are defined as 
\begin{equation}
A_I = \underset{t}{\mathrm{max}}[I(t)] - I(0) = I(t_I)-N \gamma_0,
\end{equation}
Our results, displayed in Fig.~\ref{HtHgamma2}, show that the height $A_I$ of the pulse is maximal for $\gb=0$, decreases monotonically with $\gb$ and vanishes for $\gb\geqslant\gb^{*}$, where the critical value $\gb^{*}$ depends only on the number of atoms. The decrease as a function of $\gb$ is more and more linear as $N$ increases. We explain this behavior in the next subsection on the basis of a mean-field approximation. For sufficiently large $N$, the time $t_I$ at which the maximum occurs increases as a function of $\gb$ before dropping to zero at $\gb=\gb^*$. The critical value $\gb^{*}$ increases as the number of atoms increases, as shown in Fig.~\ref{deltagammacritic}, and tends to $\gamma_0$ for $N \to \infty$. This means that for a fixed value of $\gb$, superradiance can always be observed for a sufficiently large number of atoms. Indeed, the derivative of the radiated energy rate (\ref{RER3}) reads
\begin{equation}\label{dRER3}
\frac{dI(t)}{dt} = \sum_{J = J_\mathrm{min}}^{N/2} d_{N}^{J} \sum_{M = -J}^{J} \tilde{c}_{J}^{M} \, \rho_{J}^{M,M}(t)
\end{equation}
with
\begin{multline}\label{dRER3coeff}
\tilde{c}_{J}^{M} = 2 \big(J+M\big)\big(J-M+1\big) \big[(M-1)\gamma - \gb\big] \gamma  \\ - \left(M + \frac{N}{2}\right) \gb^2.
\end{multline}
If the derivative of the radiated energy rate at initial time is strictly positive, a non-zero superradiant pulse height ($A_{I} > 0$) is always obtained. For an initial fully excited state, this sufficient condition in terms of the critical value $\Delta\gamma^{*}(N)$ reads
\begin{equation}\label{DeltaGammaStar}
\gb < \gamma_0\left(1 - \frac{1}{\sqrt{N-1}}\right) \equiv \Delta\gamma^{*}(N).
\end{equation}
As shown in Fig.~\ref{deltagammacritic}, our numerical results are in excellent agreement with Eq.~(\ref{DeltaGammaStar}).

\begin{figure}[hbt]
\begin{center}
\includegraphics[width=0.475\textwidth]{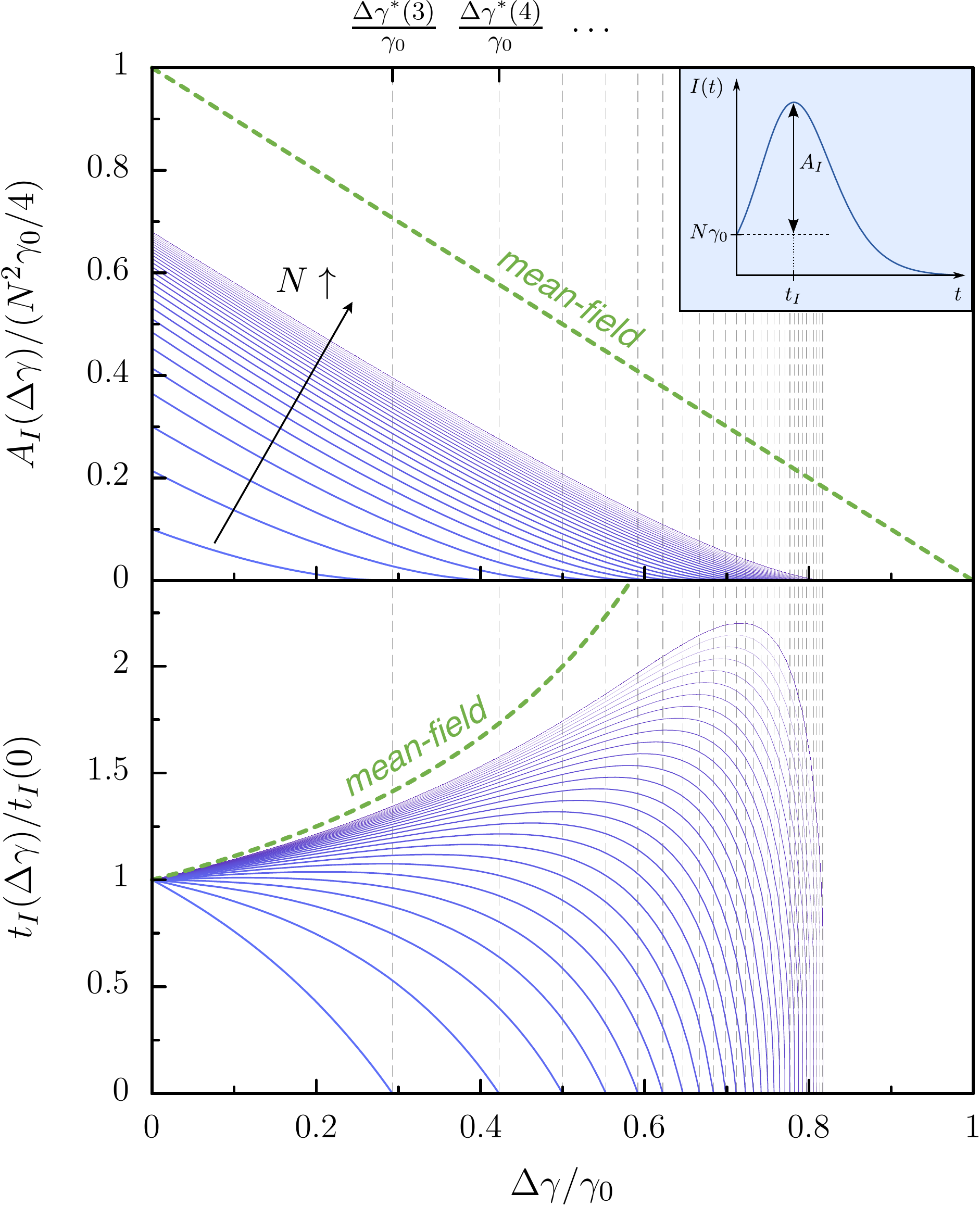}
\end{center}
\caption{Shown on top is the height $A_I$ of the superradiant pulse rescaled by $N^2\gamma_0/4$ as a function of $\gb = \gamma_0 - \gamma$ for $N = 3, \dotsc, 30$ (from left to right). The bottom shows the delay time $t_I(\gb)$ after which the radiated intensity attains a maximum, rescaled by $t_I(0)$. The dashed green curves correspond to the mean-field results [Eqs.~(\ref{Hmf}) and (\ref{thsupmf})].}\label{HtHgamma2}
\end{figure}

\begin{figure}[hbt]
\begin{center}
\includegraphics[width=0.475\textwidth]{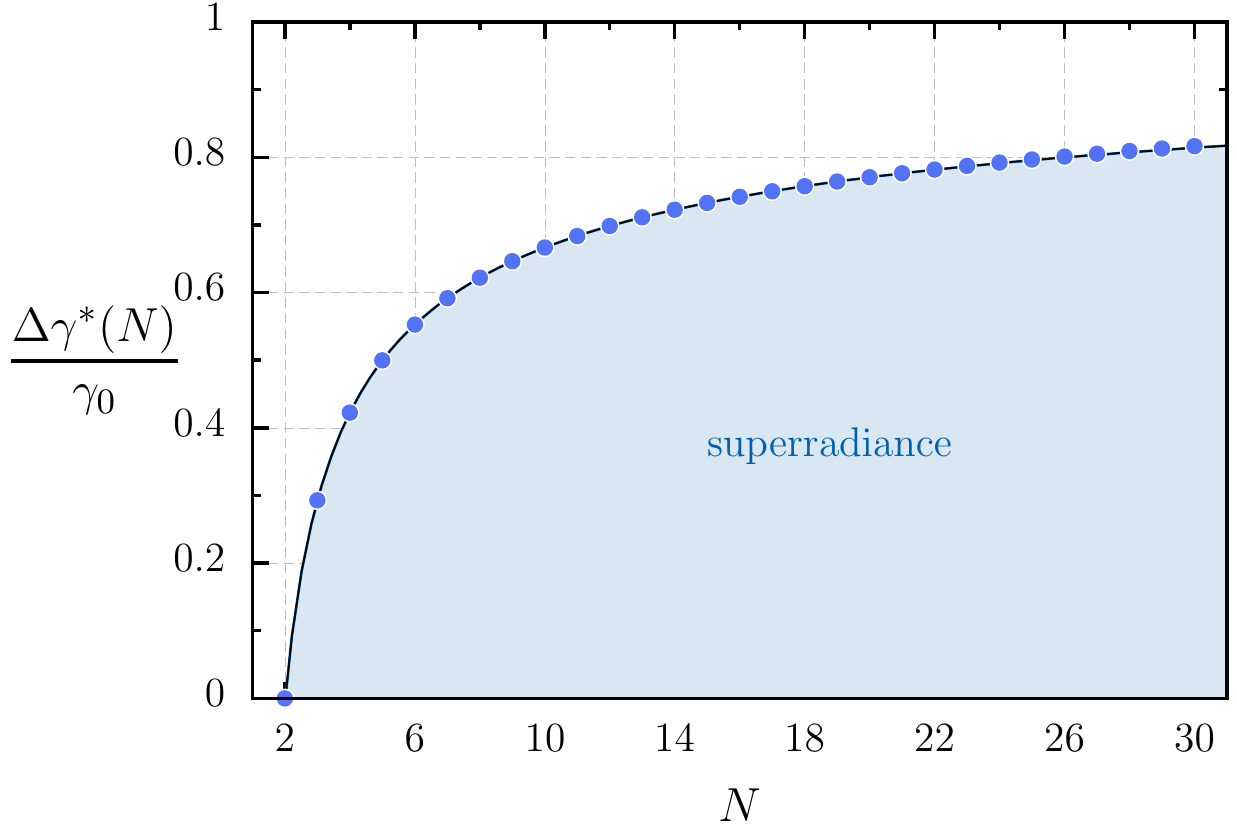}
\end{center}
\caption{Critical value $\gb^*$ at which the superradiant pulse height $A_I$ drops to zero and remains zero for $\gb>\gb^*$, plotted as a function of the number of atoms. The circles show the values extracted from numerical computations. The solid line shows the analytical prediction given by Eq.~(\ref{DeltaGammaStar}).}\label{deltagammacritic}
\end{figure}

\subsubsection{Mean field approach}

When the number of atoms is large, a mean-field approximation can be made~\cite{Bre06, Spo80} that assumes an internal state of the form
\begin{equation}
\label{rhoMeanField}
\rho(t) \approx \sigma(t) \otimes \cdots \otimes \sigma(t).
\end{equation}
In the mean-field approximation, all atoms lie in the same quantum state $\sigma(t)$. The global state $\rho(t)$
is permutation invariant at any time $t$ but not necessarily
symmetric. However, when $\sigma(t)$ is a pure state, $\rho(t)$ is
symmetric and has only components in the block of maximal angular
momentum $J=N/2$. When $\Delta \gamma=0$, the superradiant cascade takes place only in the block $J=N/2$ and $\sigma(t)$ is usually chosen pure~\cite{Bre06}. When $\Delta \gamma\ne 0$, the ratio between the transition rates within the block $J = N/2$ and the neighboring block $J = N/2 - 1$ for the emission of the $s$-th photon with $s\gg 1$ is much larger than $1$, i.e.\
\begin{equation}
\frac{\G{{N/2}^{N/2 - s + 1,N/2 - s + 1}}^{(2)} }{\G{{N/2}^{N/2 - s+1,N/2 - s+1}}^{(3)}} \approx s\frac{\gamma}{\Delta\gamma}\gg 1.
\end{equation}
Hence, during the main part of the radiative cascade (when $s$ is large), the dynamics takes place essentially in the block $J=N/2$, so that we also choose $\sigma(t)$ to be a pure state.

By inserting Eq.~(\ref{rhoMeanField}) into the master equation (\ref{meqsummary}) and by tracing over $N-1$ atoms, we get the following non-linear equation for $\sigma(t)$
\begin{equation}\label{mfeq}
\frac{d\sigma(t)}{dt} = -\frac{i}{\hbar} \Big[ V_H\left[\sigma(t) \right] + V_D\left[ \sigma(t) \right] , \sigma(t) \Big] + \mathcal{D}_\mathrm{se}\left[ \sigma(t) \right].
\end{equation}
In Eq.~(\ref{mfeq}), $V_H$ is the non-linear Hartree potential (proportional to the dipole-dipole shift $\Delta_{\mathrm{dd}}$)
\begin{equation}\label{mfVH}
V_H\left[ \sigma(t) \right] =  \hbar \Delta_{\mathrm{dd}} (N-1) \, \big(\langle\sigma_+\rangle \sigma_- +\langle\sigma_-\rangle \sigma_+  \big)
\end{equation}
where $\langle\, \boldsymbol{\cdot}\,\rangle=\mathrm{Tr}[\boldsymbol{\cdot}\,\sigma(t)]$, $V_D$ is the non-linear dissipative potential
\begin{equation}\label{mfVD}
V_D\left[ \sigma(t) \right] = i \hbar \gamma \frac{N-1}{2} \, \big(\langle\sigma_+\rangle \sigma_- - \langle\sigma_-\rangle \sigma_+  \big),
\end{equation}
and $ \mathcal{D}_\mathrm{se}$ is the single-atom dissipator accounting for spontaneous emission
\begin{equation}\label{dissipmf}
\mathcal{D}_\mathrm{se}\left[ \sigma(t) \right] =  \gamma_0 \left( \sigma_- \sigma(t) \sigma_+ - \frac{1}{2} \left\{ \sigma_+ \sigma_-, \sigma(t) \right\} \right).
\end{equation}
Equation (\ref{mfeq}) cannot be solved analytically because of the presence of the term (\ref{dissipmf}). However, as $N$ gets large, this one can be neglected in comparison to (\ref{mfVH}) and (\ref{mfVD}) provided $\gamma\ne 0$ and $N-1$ can be replaced by $N$. Equation (\ref{mfeq}) can then be related for pure states $\sigma_\psi(t)=\ket{\psi(t)}\bra{\psi(t)}$ to a non-linear Schr\"odinger equation for $\ket{\psi(t)}$ of the form (in the interaction picture) \cite{Bre06, Man95, Pra14}
\begin{equation}\label{mfschro}
\frac{d \ket{\psi(t)}}{dt} = -\frac{i}{\hbar} \big(V_H[\sigma_\psi(t)] + V_D[\sigma_\psi(t)]\big) \ket{\psi(t)}.
\end{equation}
As in \cite{Bre06}, we parametrize the state $\ket{\psi(t)}$ by
\begin{equation}\label{psidecompmf}
\ket{\psi(t)} = \sqrt{p(t)} \, e^{i \theta(t)} \ket{e} +  \sqrt{1-p(t)}\, \ket{g}
\end{equation}
with $p(t)= |\langle e | \psi(t)\rangle|^2$ the mean number of atoms in the excited state. By inserting Eq.~(\ref{psidecompmf}) into (\ref{mfschro}) we get
\begin{subequations}\label{ptthetat} 
\begin{align}
& \frac{dp(t)}{dt} = - N \gamma \, p(t) [1-p(t)], \label{ptthetata} \\[5pt]
& \frac{d\theta(t)}{dt} = - N \Delta_{\mathrm{dd}} [1-p(t)].
\end{align}
\end{subequations}
With the conditions $p(t_I) = 1/2$ and $\theta(0) = \theta_0$, the system of equations (\ref{ptthetat}) has the unique solution
\begin{align}
& p(t) = \frac{1}{1+e^{N \gamma (t-t_I)}}, \\[5pt]
& \theta(t) = \theta_0 + \frac{\Delta_{\mathrm{dd}}}{\gamma} \ln\left[p(t)\left(1 + e^{-N \gamma t_I}\right)\right],
\end{align}
where $t_I$ corresponds to the time at which half the photons have been emitted and is identified with the delay time of superradiance~\cite{Gro82}. The phase $\theta(t)$ depend on both the dipole-dipole shift $\Delta_{\mathrm{dd}}$ and the decay rate $\gamma$, while the population $p(t)$ depends only on the rate $\gamma$. In other words, the dissipative dynamics is not affected by the dipole-dipole shift. The radiated energy rate in the mean-field approximation reads
\begin{equation}\label{RERmeanfield}
I_{\mathrm{mf}}(t) = -N \frac{d p(t)}{dt} = \frac{N^2 \gamma}{4} \cosh^{-2}{\left[ \frac{N \gamma}{2} (t-t_I)\right]},
\end{equation}
and is of the same form as the pure superradiant pulse for colocated atoms~\cite{Bre06, Gro82}, but with $\gamma_0$ replaced by $\gamma$ and \textit{a priori} a different delay time $t_I$. The quantum fluctuations of the atomic positions modify the value of $\gamma$ as compared to $\gamma_0$, and thus the shape of the superradiant pulse, which is however always present except for $\gamma=0$. The height $A_{I,\mathrm{mf}}$ of the pulse (\ref{RERmeanfield}), given by
\begin{equation}\label{Hmf}
A_{I,\mathrm{mf}} = \frac{N^2 \gamma}{4} =  \frac{N^2 \gamma_0}{4}\left(1-\frac{\gb}{\gamma_0}\right),
\end{equation}
is always smaller than the height $N^2 \gamma_0/4$ of the pure superradiant pulse since $\gamma \leqslant \gamma_0$. Equation~(\ref{Hmf}) is compared with numerical simulations in Fig.~\ref{HtHgamma2} (green dashed curve, top panel). As for the delay time $t_I$, it cannot be evaluated precisely in the mean-field approach. Nevertheless, an approximation can be obtained in the limit $N\to\infty$ and for $\gamma\ne 0$ by evaluating the sum of the typical times between each photon emission~\cite{Gro82}. We find
\begin{equation}\label{thslow}
t_I \sim \frac{\ln N}{N \gamma}
\end{equation}
which corresponds to the result of Gross and Haroche~\cite{Gro82} but with $\gamma_0$ replaced by $\gamma$. The ratio between the delay time for $\gb \ne 0$ and the one for $\gb = 0$ (pure superradiance) is thus given by
\begin{equation}\label{thsupmf}
\frac{t_I(\Delta\gamma)}{t_I(0)} = \frac{\gamma_0}{\gamma} = \frac{1}{1-(\gb/\gamma_0)}.
\end{equation}
It is always larger than $1$ and increases with $\Delta \gamma$, meaning that the larger $\Delta \gamma$ is, the longer it takes before the radiated energy rate attains a maximum. Equation~(\ref{thsupmf}) is compared with numerical simulations in Fig.~\ref{HtHgamma2} (green dashed curve, bottom panel).

\subsection{Subradiance}
\label{secSubradiance}

Subradiant states are states for which the radiated energy rate decays slowly as compared to the one corresponding to independent spontaneous emission.
Dark (or decoherence-free) states are a particular class of subradiant states for which the radiated energy rate (\ref{RER3}) vanishes.  According to  Eq.~(\ref{RER3coeff}), their only non-zero populations $\rho_J^{M,M}$ are those for which $J$ and $M$ are such that $c_{J}^{M}=0$. When $\gb = 0$, the condition $c_{J}^{M}=0$ is satisfied for $M = -J$~\cite{Kar07}. As a consequence, all states $\ket{J,-J}$ (in number $\alpha_N^{J_\mathrm{min}}$; see, e.g.,~\cite{Bra01}) are dark states. When $\gb > 0$, the only dark state is obtained for $J = M = N/2$ and corresponds to the ground state $\ket{g,\dotsc,g}$. 

In the following, we study the time evolution of the state $\ket{J_0,-J_0}$ (with $J_0\in\{J_{\mathrm{min}},\dotsc,N/2\}$) when $\gb > 0$. The initial non-zero matrix element $\rho_{J_0}^{-J_0,-J_0}$ is only coupled to the matrix elements $\rho_{J}^{-J,-J}$ with higher angular momenta $J$, i.e.\ $J_0 \leqslant J\leqslant N/2$, as can be seen from Fig.~\ref{rates}. The system will thus gradually populate all states $\ket{J, -J}$ with $J>J_0$ before finally reaching the ground state $\ket{N/2, -N/2}$. The populations $\rho_{J}^{-J,-J}$ are obtained from Eq.~(\ref{diffeq}), which simplifies to
\begin{equation}\label{subEQ}
\begin{aligned}
\frac{d\rho_{J}^{-J,-J}(t)}{dt} = & -\gb \Bigg[ \, \left(\frac{N}{2} - J\right)\rho_{J}^{-J,-J}(t) \\ 
&\, - \frac{d_N^{J-1}}{d_{N}^{J}}\left( \frac{N}{2} - J+1\right) \rho_{J-1}^{-J+1,-J+1}(t) \Bigg]
\end{aligned}
\end{equation}
and admits the solution
\begin{equation}
\rho_{J}^{-J,-J}(t) = \frac{\left(\frac{N}{2}-J_0\right)!\;e^{-\gb \left( \frac{N}{2} - J_0 \right) t}}{d_{N}^{J}\left(\frac{N}{2}-J\right)!\left(J-J_0 \right)!} \left(e^{\gb \, t} - 1 \right)^{J-J_0}.
\end{equation}
Inserting this expression into Eq.~(\ref{RER3}) for the radiated energy rate yields after some calculations
\begin{equation}\label{RERsub}
I(t) = \gb \left( \frac{N}{2} - J_0 \right) e^{-\gb t}.
\end{equation}
Hence, $I(t)$ decreases exponentially regardless of the initial angular momentum $J_0$, except for the case $J_0 = N/2$ (ground state) for which $I(t)=0$ at any time $t$. We also see that the states $\ket{J_0,-J_0}$ are subradiant, since the emission rate $\gb$ is always smaller than $\gamma_0$, the single-atom spontaneous emission rate.

\section{Conclusions}
We have investigated superradiance and subradiance from
indistinguishable atoms  with quantized motional state based on the
master equation derived in \cite{Dam16}.  The indistinguishability of
the atoms implies that for an initially factorized state between the
external (center-of-mass) and internal degrees of freedom the motional
state must be invariant under permutation of atoms.  As a consequence,
the whole dynamics is parametrized only by three real numbers, namely
the diagonal $\gamma_0$ and off-diagonal $\gamma\leqslant\gamma_0$ decay rates, and a dipole-dipole shift
$\Delta_{\mathrm{dd}}$ that is identical for all atoms. All three
parameters can be ``quantum-programmed'' by appropriate choice of the
motional state of the atoms. For $\gamma=\gamma_0$ standard
superradiance results, whereas for $\gamma\to 0$ individual
spontaneous emission of the atoms prevails. A continuous transition
between these two extreme cases can be achieved. A superradiant enhancement
of the emitted intensity is always observed for
$\gamma>\gamma_0/\sqrt{N-1}$ where $N$ is the number of atoms. All non-trivial dark states (i.e.~states other
than the ground state with strictly vanishing emission of radiation)
are immediately lost as soon as $\gamma<\gamma_0$. This implies that
for harmonically trapped atoms, exact decoherence free subspaces 
that protect against spontaneous emission through destructive
interference of individual spontaneous emission amplitudes exist only in the limit of
classically localized atoms, i.e.~atoms in infinitely steep traps. Finally, we showed that the states that are dark when $\gamma=\gamma_0$ are only subradiant when $\gamma<\gamma_0$.

\begin{acknowledgments}
FD would like to thank the FRS-FNRS (Belgium) for financial support.
FD is a FRIA (Belgium) grant holder of the Fonds de la Recherche Scientifique-FNRS (Belgium).
\end{acknowledgments}

\section*{Appendix A : Symmetry of density matrix under permutation of indistinguishable atoms}

In this Appendix, we give general properties under permutation of atoms of the (reduced) density matrices describing the states of indistinguishable atoms. 

Consider a set of $N$ indistinguishable atoms (bosons or fermions) with internal and external degrees of freedom. We define the orthonormal basis vectors as
$\ket{\bnu}\ket{\bphi}\equiv\ket{\nu_1,\ldots,\nu_N}\ket{\phi_1,\ldots,
  \phi_N}$, where $\ket{\nu_j}$ (resp.\ $\ket{\phi_j}$) are the
internal (resp.\ external) orthonormal basis states of the particle $j$. The permutation
operator $P_\pi$ corresponding to 
the permutation $\pi$ is defined through exchange of the particle labels
in the basis states, i.e.
\begin{equation}
  \label{eq:P}
  P_\pi\ket{\bnu}\ket{\bphi}=\ket{\nu_{\pi_1},\ldots,\nu_{\pi_N}}\ket{\phi_{\pi_1},\ldots,
    \phi_{\pi_N}}\equiv \ket{\bnu_\pi}\ket{\bphi_\pi}\,.
\end{equation}
We have $P_\pi = P_\pi^\mathrm{in}\otimes P_\pi^\mathrm{ex}$, where $P_\pi^\mathrm{in}$ and $P_\pi^\mathrm{ex}$ are such that $P_\pi^\mathrm{in} \ket{\bnu} = \ket{\bnu_\pi}$ and $P_\pi^\mathrm{ex} \ket{\bphi} = \ket{\bphi_\pi}$.

An arbitrary pure state $\ket{\psi}$ of the full system can be written as
\begin{equation}
  \label{eq:psifull}
  \ket{\psi}=\sum_{\bnu\bphi}\alpha_{\bnu\bphi}\ket{\bnu}\ket{\bphi}\,
\end{equation}
and must be invariant under permutations up to a global phase, i.e.\
$P_\pi\ket{\psi}=(\pm)^{p_\pi}\ket{\psi}$, where $p_\pi$ is the parity
of the permutation (even 
or odd), and $(\pm)^{p_\pi}$ the 
phase factor picked up accordingly for bosons ($+$) or fermions
($-$). 
Then we have the following:   
\begin{lemma}
An arbitrary mixed state $\rho$ of indistinguishable {\em bosons} or {\em fermions}  (density
operator on the full Hilbert space) satisfies 
\begin{equation}
  \label{eq:full}
  P_\pi\rho P_{\pi'}^\dagger=(\pm)^{p_\pi + p_{\pi'}}\rho \;\; \forall \pi,\,\pi'\,.
\end{equation}
\end{lemma}
\begin{proof}
The mixed state of a system of indistinguishable bosons (fermions) must be 
a mixture of pure states that have all the full permutation symmetry (antisymmetry), i.e.~
\begin{equation}
  \label{eq:rho}
  \rho=\sum_i p_i \ketbra{\psi^{(i)}}{\psi^{(i)}}
\end{equation}
where $p_i$ are probabilities and $P_\pi\ket{\psi^{(i)}}=(\pm)^{p_\pi}\ket{\psi^{(i)}}$ for all $i$. Applying $P_\pi$ from the left and $P_{\pi'}^\dagger$ from the
right immediately yields the claim. 
\end{proof}

Consider now the reduced density matrix corresponding to the internal
degrees of freedom only. Inserting the decomposition \eqref{eq:psifull}
for each state $\ket{\psi^{(i)}}$ in the convex sum \eqref{eq:rho},
we obtain
\begin{equation}
  \rho^\mathrm{in}\equiv\tr_{\rm
    ex}\rho=\sum_{\bphi}\bra{\bphi}\rho\ket{\bphi} = \sum_ip_i\sum_{\bphi,\bnu,\bmu}\alpha_{\bnu\bphi}^{(i)}\alpha_{\bmu\bphi}^{(i)*}\ket{\bnu}\bra{\bmu}\,.   \label{eq:rhoin}
\end{equation}
Similarly, the reduced density matrix corresponding to the external degrees of freedom reads 
\begin{equation}
\rho^\mathrm{ex}\equiv\tr_{\rm
    in}\rho=\sum_{\bnu}\bra{\bnu}\rho\ket{\bnu} = \sum_ip_i\sum_{\bnu,\bphi,\bpsi}\alpha_{\bnu\bphi}^{(i)}\alpha_{\bnu\bpsi}^{(i)*}\ket{\bphi}\bra{\bpsi}\,.  
\end{equation}

Then we have the following
\begin{lemma}
The arbitrary reduced density matrices $\rho^\mathrm{in}$ and $\rho^\mathrm{ex}$ of indistinguishable
atoms (bosons or fermions) satisfy 
\begin{equation}
\begin{aligned}
  \label{eq:bos}
  &P^\mathrm{in}_\pi\rho^\mathrm{in} P_{\pi}^{\mathrm{in}\dagger}=\rho^\mathrm{in}  \quad\forall\, \pi\,, \\
  &P^\mathrm{ex}_\pi\rho^\mathrm{ex}  P_{\pi}^{\mathrm{ex}\dagger}=\rho^\mathrm{ex} \quad\forall\, \pi\,.
  \end{aligned}
\end{equation}
\end{lemma}

\begin{proof}
We present here the proof for $\rho^\mathrm{in}$. The symmetry of the full state implies the symmetry of the
coefficients $\alpha_{\bnu\bphi}^{(i)}$:
\begin{eqnarray}
      \label{eq:syma}
P_\pi\ket{\psi^{(i)}}&=&\sum_{\bnu\bphi}\alpha_{\bnu\bphi}^{(i)}\ket{\bnu_\pi}\ket{\bphi_\pi}\\
&=&\sum_{\bnu\bphi}\alpha_{\bnu_{\pi^{-1}}\bphi_{\pi^{-1}}}^{(i)}\ket{\bnu}\ket{\bphi}  =(\pm)^{p_\pi}\ket{\psi^{(i)}}\,,            
\end{eqnarray}
and projecting onto the basis states gives
\begin{equation}
  \label{eq:as}
  (\pm)^{p_\pi}\alpha^{(i)}_{\bnu\bphi}=\alpha^{(i)}_{\bnu_{\pi^{-1}}\bphi_{\pi^{-1}}}\,.
\end{equation}
As a consequence,
\begin{equation}
\begin{aligned}
  \label{eq:symrhoin}
  P_\pi^\mathrm{in}\rho^\mathrm{in}  P_\pi^{\mathrm{in}\dagger}&= \sum_{i, \bphi,\bnu,\bmu}p_i\alpha^{(i)}_{\bnu\bphi}\alpha^{(i)*}_{\bmu\bphi}\ket{\bnu_\pi}\bra{\bmu_\pi}\\
&=\sum_{i, \bphi,\bnu,\bmu}p_i\alpha^{(i)}_{\bnu_{\pi^{-1}}\bphi}\alpha^{(i)*}_{\bmu_{\pi^{-1}}\bphi}\ket{\bnu}\bra{\bmu}
\\
&=\sum_{i, \bphi,\bnu,\bmu}p_i\alpha^{(i)}_{\bnu_{\pi^{-1}}\bphi_{\pi^{-1}}}\alpha^{(i)*}_{\bmu_{\pi^{-1}}\bphi_{\pi^{-1}}}\ket{\bnu}\bra{\bmu}\nonumber  =\rho^{\mathrm{in}}\,,
\end{aligned}
\end{equation}
where in the penultimate step permutation $\pi$ was
absorbed in the sum over all $\bphi$, and the last step follows from
Eqs.~(\ref{eq:as}) and (\ref{eq:rhoin}).
\end{proof}
Note that in general for the reduced density matrix $\rho^\mathrm{in}$ the statement
corresponding to Eq.~(\ref{eq:full}) does not
hold, i.e.~$P_\pi\rho^\mathrm{in} P_{\pi'}^\dagger\ne\rho^\mathrm{in}$  for
$\pi\ne\pi'$:  Going through the last proof again with the second
$\pi$ 
replaced by $\pi'$, one realizes that in at least one of
the coefficients $\alpha^{(i)}_{\bnu_{\pi^{-1}}\bphi}$ or
$\alpha^{(i)*}_{\bmu_{\pi^{'-1}}\bphi}$, $\bphi$ cannot be replaced by
  $\bphi_{\pi^{-1}}$ or $\bphi_{\pi^{'-1}}$ if $\pi\ne \pi'$, and in
  general $\alpha_{\bnu\bphi}\ne\alpha_{\bnu_{\pi^{-1}}\bphi_{\pi^{'-1}}}$
  even for bosons.

\section*{Appendix B : General expressions of decay rates and dipole-dipole shifts}

In this Appendix, we show that all
off-diagonal ($i\ne j$) decay rates $\gamma_{ij}$ and all dipole-dipole shifts $\Delta_{ij}$ are equal for any pair of indistinguishable atoms $i$ and $j$ in arbitrary permutation invariant motional states. Then, we give their general expressions for arbitrary symmetric or antisymmetric motional states.

As shown in~\cite{Dam16}, the diagonal decay rates are equal to the single-atom spontaneous emission rate $\gamma_0$ for any motional state while the off-diagonal decay rates and dipole-dipole shifts are given, respectively, by
\begin{align}\label{gammaijconv}
\gamma_{ij} &= \int_{\mathbb{R}^3} \gamma^{\mathrm{cl}}(\mathbf{r}) \,\mathcal{F}^{-1}_{\mathbf{r}} \left[\corkmz\right] d\mathbf{r}, \\
\label{deltaijconv2}
\Delta_{ij} &= \int_{\mathbb{R}^3} \Delta^{\mathrm{cl}}(\mathbf{r}) \,\mathcal{F}^{-1}_{\mathbf{r}} \left[\corkmz\right] d\mathbf{r},
\end{align}
with $\mathcal{F}^{-1}_{\mathbf{r}} \left[\corkmz\right]$ the inverse Fourier transform of the motional correlation function~\cite{footnote}
\begin{equation}\label{cijex}
\corkmz = \mathrm{Tr}_{\mathrm{ex}}\left[ e^{i\mathbf{k}\boldsymbol{\cdot} \hat{\mathbf{r}}_{ij}} \rho_A^\mathrm{ex} \right],
\end{equation}
where $\hat{\mathbf{r}}_{ij} = \hat{\mathbf{r}}_{i} - \hat{\mathbf{r}}_j$ is the difference between the position operators of atoms $i$ and $j$. In Eqs.~(\ref{gammaijconv}) and (\ref{deltaijconv2}), $\gamma^{\mathrm{cl}}(\mathbf{r}) $ and
$\Delta^{\mathrm{cl}}(\mathbf{r})$ are the classical expressions of the
decay rates and dipole-dipole shifts, respectively,
for a pair of atoms connected by $\mathbf{r}$ and a radiation of wavenumber $k_0$~\cite{Ste64,Leh70,Aga74},
\begin{equation}
\label{gammaijcl}
\gamma^{\mathrm{cl}}(\mathbf{r}) =\frac{3  \gamma_0 }{2}\Bigg[ p \,\frac{\sin (k_0 r)}{k_0 r} + q  \left( \frac{\cos(k_0 r)}{(k_0 r)^2} - \frac{\sin(k_0 r)}{(k_0 r)^3}\right)\Bigg]
\end{equation}
and
\begin{equation}
\label{deltaijcl}
\Delta^{\mathrm{cl}}(\mathbf{r}) =\frac{3  \gamma_0 }{4}\Bigg[ - p \,\frac{\cos(k_0 r)}{k_0 r} + q  \left( \frac{\sin(k_0 r)}{(k_0 r)^2} + \frac{\cos(k_0 r)}{(k_0 r)^3}\right) \Bigg].
\end{equation} 
with $p$ and $q$ angular factors given by
\begin{equation}\label{p}
p=\begin{cases}
\sin^2 \alpha & \mbox{for a $\pi$ transition}\\
\tfrac{1}{2}(1+\cos^2 \alpha) & \mbox{for a $\sigma^\pm$ transition}
\end{cases}
\end{equation} 
and
\begin{equation}\label{q}
q=\begin{cases}
1-3 \cos^2 \alpha & \mbox{for a $\pi$ transition}\\
\tfrac{1}{2}(3 \cos^2 \alpha-1) & \mbox{for a $\sigma^\pm$
  transition,}
\end{cases}
\end{equation} 
where $\alpha=\arccos(\mathbf{e}_r\boldsymbol{\cdot}\mathbf{e}_z)$ is the angle between the quantization axis and $\mathbf{r}$.

Indistinguishability of atoms implies that their motional state is invariant under permutation [see Appendix A], i.e.\
\begin{equation}
P^\mathrm{ex}_\pi\rho_A^\mathrm{ex}  P_{\pi}^{\mathrm{ex}\dagger}=\rho_A^\mathrm{ex} \quad\forall\, \pi\,.
\end{equation}
Upon using the latter equation, the motional correlation function (\ref{cijex}) is found to satisfy
\begin{equation}\label{equality}
\begin{aligned}
\corkmz &= \mathrm{Tr}_{\mathrm{ex}}\left[ e^{i\mathbf{k}\boldsymbol{\cdot} \hat{\mathbf{r}}_{ij}} P^\mathrm{ex}_\pi \rho_A^\mathrm{ex} P^{\mathrm{ex}\dagger}_\pi\right] \\
&= \mathrm{Tr}_{\mathrm{ex}}\left[P^{\mathrm{ex}\dagger}_\pi e^{i\mathbf{k}\boldsymbol{\cdot} \hat{\mathbf{r}}_{ij}} P^\mathrm{ex}_\pi \rho_A^\mathrm{ex} \right] \\
&= \mathrm{Tr}_{\mathrm{ex}}\left[ e^{i\mathbf{k}\boldsymbol{\cdot} \hat{\mathbf{r}}_{\pi(i)\pi(j)}} \rho_A^\mathrm{ex} \right] = \mathcal{C}_{\pi(i)\pi(j)}^\mathrm{ex}(\mathbf{k}).
\end{aligned}
\end{equation}
The equality of $\corkmz$ for any pair of atoms [Eq.~(\ref{equality})] implies the equality of the decay rates (\ref{gammaijconv}) [or the dipole-dipole shifts (\ref{deltaijconv2})] for any pair of atoms.

Consider now an arbitrary symmetric or antisymmetric motional state of the form
\begin{equation}\label{NatMSmix}
\rho_A^{\mathrm{ex},\pm} = \sum_{m = 1}^M p_m \big|\Phi_A^{(m),\pm}\big\rangle \big\langle\Phi_A^{(m),\pm}\big|,
\end{equation}
where $p_m$ are the weights of the statistical mixture ($p_m \geq 0$ and $\sum_m p_m = 1$)
and $\big|\Phi_A^{(m),\pm}\big\rangle$ ($m = 1, \dotsc, M$) are symmetric ($+$) or antisymmetric ($-$) $N$-atom motional pure states. Any state $\big|\Phi_A^{(m),\pm}\big\rangle$ can be written as
\begin{equation}\label{NatMS}
\begin{aligned}
\big|\Phi_A^{(m),\pm}\big\rangle=\sqrt{\frac{n_{\phi_1^{(m)}}!\cdots n_{\phi_N^{(m)}}!}{N!}}\,\sum_{\pi}  (\pm 1)^{p_\pi} \,
\big|\phi_{\pi(1)}^{(m)}
\cdots \phi_{\pi(N)}^{(m)}\big\rangle
\end{aligned}
\end{equation}
where $\big|\phi_{j}^{(m)}\big\rangle$ ($j = 1, \dotsc, N$) are normalized (but not necessarily orthogonal) single-atom motional states, $n_{\phi_j^{(m)}}$ is the number of atoms occupying the state $\big|\phi_j^{(m)}\big\rangle$, and the sum runs over all permutations $\pi$ of the atoms.

The off-diagonal decay rates and the dipole-dipole shifts for the motional state (\ref{NatMSmix}) can be expressed in terms of exchange integrals as~\cite{Dam16}
\begin{widetext}
\begin{align}
&\gamma_{ij} = \sum_{m = 1}^{M} p_m \sum_{\pi, \pi'} \lambda_{ij,\pi\pi'}^{(m),\pm} \iint_{\mathbb{R}^3\times \mathbb{R}^3} \gamma^{\mathrm{cl}}(\mathbf{r}-\mathbf{r}') \,\phi_{\pi(i)}^{(m)}(\mathbf{r}) \, \phi_{\pi'(i)}^{(m)*}(\mathbf{r}) \, \phi_{\pi(j)}^{(m)}(\mathbf{r}') \,\phi_{\pi'(j)}^{(m)*}(\mathbf{r}')
 \,d\mathbf{r}\, d\mathbf{r}', \label{indisEXCHANGE1} \\
&\Delta_{ij} =  \sum_{m = 1}^{M} p_m  \sum_{\pi, \pi'} \lambda_{ij,\pi\pi'}^{(m),\pm} \iint_{\mathbb{R}^3\times \mathbb{R}^3} \Delta^{\mathrm{cl}}(\mathbf{r}-\mathbf{r}') \,\phi_{\pi(i)}^{(m)}(\mathbf{r})\, \phi_{\pi'(i)}^{(m)*}(\mathbf{r}) \, \phi_{\pi(j)}^{(m)}(\mathbf{r}') \, \phi_{\pi'(j)}^{(m)*}(\mathbf{r}')
 \,d\mathbf{r}\, d\mathbf{r}',\label{indisEXCHANGE2}
\end{align}
\end{widetext}
with $\phi_j^{(m)}(\mathbf{r}) = \big\langle \mathbf{r} | \phi_j^{(m)}\big\rangle$ the single-atom motional states in the position representation, 
\begin{equation}\label{pij}
 \lambda_{ij,\pi\pi'}^{(m),\pm} = \frac{\displaystyle  (\pm 1)^{p_\pi + p_{\pi'}}  \, \prod_{n = 1 \atop n\neq i,j}^N \big\langle\phi_{\pi'(n)}^{(m)}\big|\phi_{\pi(n)}^{(m)}\big\rangle}{\displaystyle \sum_{\tilde{\pi},\tilde{\pi}'}  (\pm 1)^{p_{\tilde{\pi}} + p_{\tilde{\pi}'}} \prod_{n = 1}^N\big\langle\phi_{\tilde{\pi}'(n)}^{(m)}\big|\phi_{\tilde{\pi}(n)}^{(m)}\big\rangle}.
\end{equation} 

The cooperative decay rates and dipole-dipole shifts (\ref{indisEXCHANGE1}) and (\ref{indisEXCHANGE2}) depend on their classical expressions (\ref{gammaijcl}) and (\ref{deltaijcl}), which oscillate and decrease as a function of the interatomic distance on a length scale of the order of the wavelength of the emitted radiation. In addition, they  depend on the single-atom wavepackets and can vary as a function of their extensions and overlaps. The indistinguishability of atoms is reflected by the summations over all permutations of the atoms, which implies the equality of all off-diagonal decay rates $\gamma_{ij}$ and all dipole-dipole shifts $\Delta_{ij}$.

Note that when all atoms occupy the same motional state $\rho_1$ with spatial density $\rho_1(\mathbf{r})=\langle\mathbf{r}|\rho_1|\mathbf{r}\rangle$, the global motional state $\rho_A^\mathrm{ex}=\rho_1^{\otimes N}$ is symmetric and separable and the decay rates (\ref{indisEXCHANGE1}) and dipole-dipole shifts (\ref{indisEXCHANGE2}) merely read
\begin{equation}\label{beEXCHANGE1}
\gamma_{ij} = \iint_{\mathbb{R}^3\times \mathbb{R}^3} \gamma^{\mathrm{cl}}(\mathbf{r}-\mathbf{r}') \:\rho_1(\mathbf{r}) \:\rho_1(\mathbf{r}') \,d\mathbf{r}\, d\mathbf{r}',
\end{equation}
\begin{equation}\label{beEXCHANGE2}
\Delta_{ij} = \iint_{\mathbb{R}^3\times \mathbb{R}^3} \Delta^{\mathrm{cl}}(\mathbf{r}-\mathbf{r}') \:\rho_1(\mathbf{r}) \:\rho_1(\mathbf{r}')  \,d\mathbf{r}\, d\mathbf{r}'.
\end{equation}

\section*{Appendix C : General solution for $2$ atoms}

In this Appendix, we give the most general solution of the master equation (\ref{meqsummary}) for $N =2 $ atoms. In this case, $J = 0,1$ and the decomposition (\ref{Hdecomp}) of the internal Hilbert space of the atomic system reads
\begin{equation}\label{HdecompN2}
\mathcal{H} = \mathbb{C}^2 \otimes \mathbb{C}^2  \simeq \left(\mathcal{H}_0 \otimes \mathcal{K}_0 \right) \oplus \left(\mathcal{H}_1 \otimes \mathcal{K}_1\right),
\end{equation}
where the dimensions of $\mathcal{K}_0$ and $\mathcal{K}_1$ are $d_0 = d_1 = 1$. The value $J = 1$ defines the triplet states $\left\{\ket{1,1}, \ket{1,0}, \ket{1,-1}\right\}$ which are all symmetric while the value $J = 0$ corresponds to the singlet state $\ket{0,0}$, which is antisymmetric. In the standard basis $\left\{\ket{e,e},\ket{e,g},\ket{g,e},\ket{g,g}\right\}$, they read
\begin{equation}
\begin{array}{l} \ket{1,1} = \ket{e,e}, \\[6pt]
\displaystyle \ket{1,0} = \frac{\ket{e,g}+\ket{g,e}}{\sqrt{2}}, \\[10pt]
\ket{1,-1} = \ket{g,g}.
\end{array}
\quad
\ket{0,0} = \frac{\ket{e,g}-\ket{g,e}}{\sqrt{2}},
\end{equation}
The solutions of (\ref{diffeq}) for the density matrix elements $\rho_{J}^{M,M'}(t)$ in terms of $\gamma, \gb$ and $\Delta_{\mathrm{dd}}$ are in this case given by
\begin{equation}\label{fullsolutionN2}
\begin{aligned}
&\rho_{1}^{1,1}(t) = \rho_{1}^{1,1}(0)\, e^{-2(\gamma +\gb)t}, \\[5pt]
&\rho_{1}^{0,0}(t) = \rho_{1}^{0,0}(0)\, e^{-(2 \gamma + \gb) t} + \frac{2 \gamma + \gb}{\gb} \rho_{1}^{1,1}(t) \left(e^{\gb t} -1\right), \\[5pt]
&\rho_{1}^{-1,-1}(t) = 1 - \rho_{1}^{1,1}(t)- \rho_{1}^{0,0}(t) - \rho_{0}^{0,0}(t), \\[5pt]
&\rho_{0}^{0,0}(t) = \rho_{0}^{0,0}(0)\, e^{-\gb t} + \frac{\gb}{2 \gamma + \gb} \rho_{1}^{1,1}(t) \left(e^{(2 \gamma + \gb) t} -1\right), \\[5pt]
&\rho_{1}^{1,0}(t) = \rho_{1}^{1,0}(0) e^{-(4\gamma + 3\gb + 2i \Delta_{\mathrm{dd}})t/2}, \\[5pt]
&\rho_{1}^{1,-1}(t) = \rho_{1}^{1,-1}(0) e^{-(\gamma + \gb)t},\\[5pt]
&\rho_{1}^{0,-1}(t) = \rho_{1}^{0,-1}(0) e^{-(2 \gamma + \gb - 2 i \Delta_{\mathrm{dd}}) t/2} \\
&\hspace{0.5cm}+ \rho_{1}^{1,0}(t) \frac{2 \gamma + \gb}{\gamma + \gb + 2 i \Delta_{\mathrm{dd}}} \left(e^{( \gamma + \gb  + 2i \Delta_{\mathrm{dd}})t} - 1 \right).
\end{aligned}
\end{equation}


\begin{thebibliography}{100}
\bibitem{footnote} In this work, we use for the Fourier transform $\mathcal{F}_\mathbf{k}\left[ \cdot\right]$ and its inverse $\mathcal{F}_\mathbf{r}^{-1}[\cdot]$ the convention
\begin{equation*}
\begin{aligned}
&\ftq{f} = \int_{\mathbb{R}^3}  e^{-i \mathbf{k} \boldsymbol{\cdot} \mathbf{r}} \, f(\mathbf{r})\, d\mathbf{r}, \\
&\ftxiij{g} = \int_{\mathbb{R}^3}  e^{i  \mathbf{k}\boldsymbol{\cdot} \mathbf{r}} \,g(\mathbf{k})\, \frac{d\mathbf{k}}{(2\pi)^3}. \\
\end{aligned}
\end{equation*}
\end{thebibliography}
\end{document}